\documentclass[11pt]{article}
\usepackage[utf8]{inputenc}
\PassOptionsToPackage{hyphens}{url}\usepackage{hyperref}
\usepackage{amssymb}
\usepackage{amsthm}
\usepackage{mathtools}
\usepackage{graphicx}
\usepackage{xcolor}
\usepackage{wasysym}
\usepackage{mathrsfs}
\usepackage[english]{babel}
\usepackage[affil-it]{authblk}

\usepackage[toc,page]{appendix}

\usepackage{enumitem}
\usepackage{pifont}
\newlist{todolist}{itemize}{2}
\setlist[todolist]{label=$\square$}

\usepackage[portrait, margin=1in]{geometry}

\usepackage[nameinlink,noabbrev]{cleveref}

\DeclarePairedDelimiter\ket{\lvert}{\rangle}
\DeclarePairedDelimiterX\braket[2]{\langle}{\rangle}{#1 \delimsize\vert #2}

\newtheorem{theorem}{Theorem}[section]
\newtheorem{corollary}{Corollary}[theorem]
\newtheorem{conjecture}{Conjecture}[theorem]
\newtheorem{lemma}[theorem]{Lemma}
\theoremstyle{definition}
\newtheorem{definition}{Definition}[section]

\title{Quantum Voting and Violation of Gibbard-Satterthwaite's Impossibility Theorem}
\date{September 2023}

\author[1]{Aidan Casey}
\author[1]{Ethan Dickey \thanks{Corresponding Author: Ethan Dickey, dickeye@purdue.edu}}

\affil[1]{Department of Computer Science, Purdue University, West Lafayette, IN 47907, USA}

\begin{document}

\maketitle

\begin{abstract}

    In the realm of algorithmic economics, voting systems are evaluated and compared by examining the properties or axioms they satisfy. While this pursuit has yielded valuable insights, it has also led to seminal impossibility results such as Arrow's and Gibbard-Satterthwaite's Impossibility Theorems, which pose challenges in designing ideal voting systems. Enter the domain of quantum computing: recent advancements have introduced the concept of quantum voting systems, which have many potential applications including in security and blockchain. Building on recent works that bypass Arrow's Impossibility Theorem using quantum voting systems, our research extends Quantum Condorcet Voting (QCV) to counter the Gibbard-Satterthwaite Impossibility Theorem in a quantum setting. To show this, we introduce a quantum-specific notion of truthfulness, extend ideas like incentive compatibility and the purpose of onto to the quantum domain, and introduce new tools to map social welfare functions to social choice functions in this domain.



\end{abstract}

\newpage
\tableofcontents

\newpage
\section{Introduction}
    Quantum computing has ushered in a new era in various computational disciplines, from cryptography \cite{wiesner1983conjugate, Bennett1985Crypto, Gisin2002Crypto, MikeAndIke2011} and blockchain technology \cite{Schneier_1994, Bernstein2009, Kiktenko_2018, sun2018quantum, Sun2019QBlockchain} to complex simulations in chemistry \cite{olson2017quantum, Preskill2018quantumcomputingNISQera}, and more \cite{flitney2002introduction, HANAUSKE2007650,      eisert1999quantum, burke2020quantum}. Similarly, voting systems have expanded from their rudimentary forms in ancient civilizations to sophisticated algorithms, deeply embedded in game theory and social choice functions. This paper explores the intersection of these two influential domains, particularly focusing on the development and analysis of quantum voting systems.

    Classically, there has been much study into voting systems that take a linear order over the possible alternatives from each of the $n$ voters to determine a single winner. This is usually written in the shorthand $\mathcal{L}^n \rightarrow A$, where $\mathcal{L} = \mathcal{L}(A)$ is the set of all linear orders over the set of alternatives $A$. Voting rules that fall under this category are known as ordinal social choice functions, and are subject to the Gibbard-Satterthwaite Impossibility Theorem (GS) \cite{GibbardSatterthwaite1, GibbardSatterthwaite2}, which states that these voting systems which have three properties, (1) non-duplicity (more than 2 alternatives are possible outputs), (2) onto (every alternative could be elected), and (3) incentive compatibility (it is always just as optimal for a voter to vote truthfully), are dictatorships. Further work first by Gibbard \cite{Gibbard1977Random} and later expounded upon by Procaccia and many others \cite{Procaccia_2010, brandl2015incentives, Aziz_Luo_Rizkallah_2018, Chakrabarty2014welfare, birrell2011approximately, brandt2017rolling} deal with voting systems that produce probability mixtures over alternatives instead of a single winner. This is denoted $\mathcal{L}^n \rightarrow \Delta(A)$. Such a probability mixture can then be sampled to obtain a single winner. Again, we find these systems subject to Gibbard's Theorem (1977) \cite{Gibbard1977Random}, which states that voting systems which produce probability \textit{mixtures} over alternatives are also subject to the same impossibility conditions as non-probabilistic voting systems from GS. These impossibility theorems show fundamental restrictions when voting on non-binary decisions, which has pervasive applicability in many areas.

    Beyond voting systems that produce a single winner, the natural next step is to discuss ones that produce rankings over alternatives, denoted $\mathcal{L}^n \rightarrow \mathcal{L}$. These systems have similarly been very well studied and are subject to Arrow's Impossibility Theorem \cite{ArrowsTheorem} which, similar to GS, states that these systems which have three properties, namely (1) non-duplicity, (2) independence of irrelevant alternatives (the relative ranking of $a$ and $b$ does not depend on the relative ranking of $a$ and any other alternative), and (3) unanimity (if everyone votes for $a$ over $b$, $a$ will be over $b$), are dictatorships. Arrow's theorem has fundamental impacts across many fields, but does not fundamentally deal with truthfulness as GS does. The last non-probabilistic input voting system is, naturally, voting systems that map voter preferences in the form of rankings to probability mixtures over rankings, $\mathcal{L}^n \rightarrow \Delta(\mathcal{L})$. Barbera \cite{barbera_1978} generalizes Arrow's Theorem into this context.

    Shifting the focus to probabilistic inputs allows us to explore the specific contexts of our problem, $\Delta^n(\mathcal{L}) \rightarrow \Delta(\mathcal{L})$ and $\Delta^n(\mathcal{L}) \rightarrow \Delta(A)$. In the first context, prior research by Bao and Halpern \cite{bao2017quantumarrows} and Sun \textit{et al.} \cite{sun2021schrodinger} has shown that a version of Arrow's Theorem reformulated for this context can be bypassed. Both studies leverage the theoretical framework of quantum computing to prove their results, highlighting the potential of quantum frameworks in addressing complex problems in this field.

    Quantum computing offers an enticing solution to bypass classical restrictions. Classically, $\Delta(\mathcal{L})$ has exponential support, posing challenges for efficient manipulation. One of quantum computing's strengths lies in its ability to efficiently manipulate certain aspects of joint probability distributions over qubits, a feature central to the class of problems known as \textsc{SampBQP} - the class of sampling problems efficiently solvable on a quantum computer \cite{Lund2017}. This is particularly significant given that it is thought that sampling from such distributions classically requires exponentially scaling resources \cite{Boixo2018, Aaronson_2005, Bremner_Jozsa_Shepherd_2010, Bouland2019, aaronson2011ComplexityLinearOptics, Fujii_Morimae_2017, Bremner2016AvgCase}. This line of thinking has also been the basis of at least one quantum supremacy demonstration on the basis of random quantum circuit sampling \cite{Arute2019}. Beyond these computational advantages, quantum frameworks also enable richer ballot structures through superposition and entanglement, potentially leading to more equitable voting outcomes. Furthermore, there are devices that implement this richer context natively which are, by the definition of being quantum, more secure \cite{Gisin2002Crypto}. Further exploration of the interplay between quantum computing and game theory is detailed in Appendix \ref{app_lit_review}.

    Building upon the foundational work of Bao and Halpern \cite{bao2017quantumarrows} and Sun \textit{et al.} \cite{sun2021schrodinger}, our study seeks to answer the following question: \textit{what does truthfulness entail in quantum voting systems, and is the GS impossibility theorem applicable in a quantum context?}

    \subsection{Our approach and results.}
        In this work, we propose some notions of truthfulness and prove that we can bypass a resulting formulation of GS in a quantum setting. As a quick overview, these results are shown as follows. We first reformulate incentive compatibility/strategic voting and the GS Theorem in a quantum setting, extending the formalism introduced by Bao and Halpern, and Sun \textit{et al.} \cite{bao2017quantumarrows, sun2021schrodinger}. We then show that Quantum Condorcet Voting (QCV) is incentive compatible and onto (as it was previously shown that it is not a dictatorship \cite{sun2021schrodinger}) and that a quantum social choice function can be formed from QCV that bypasses a quantum GS theorem (i.e. forming a function in $\Delta^n(\mathcal{L}) \rightarrow \Delta(A)$ from one in $\Delta^n(\mathcal{L}) \rightarrow \Delta(\mathcal{L})$).

        A definition of GS in a quantum setting requires quantum analogues of social choice functions and incentive compatibility on those social choice functions, along with a few simpler definitions such as extending the idea of dictatorship to quantum social choice functions and notions of what onto means in such a case. With these definitions, a quantum version of GS would be defined as:
             
        \begin{conjecture}[Quantum Gibbard-Satterthwaite Conjecture (QGS)]\label{QGS}
            Every quantum incentive compatible (QIC) quantum social choice function (QSCF) $\xi$ onto alternatives $A$ with $|A| > 2$ is a quantum dictatorship.
        \end{conjecture}
    
        Next, we show that there is a Quantum Social Choice Function (QSCF) that satisfies or does not satisfy each property as required. We build upon the the previous success of Sun \textit{et al.} \cite{sun2018quantum} by adapting QCV, which they use to violate a quantum Arrow's theorem. To do so, we create the idea of Quantum Social Choice Extensions (QSCEs), which are tools to utilize Quantum Social Welfare Functions (QSWFs) to create Quantum Social Choice Functions, and a particular extension that allows for some of the useful properties of the welfare function to be maintained.
    
        For clarity, we begin by defining the idea of QSCFs (\Cref{def_QSCF}, QSCEs (\Cref{def_QSCE}), and the particular QSCE we utilize. This QSCE, which we call the Natural Quantum Social Choice Extension (NQSCE, \Cref{def_NQSCE}), is called as such because it essentially maps probability of rankings to the probability of their highest ranked alternative. Once these are defined, it allows for each property on QSCFs to be defined and proven (or disproven) when applied to the composition of QCV and the NQSCE which we call QCV with the Natural Extension (QCVNE).
    
        The first property definition is Quantum Incentive Compatibility, which is the most involved. First, the notion of quantum preference is defined (\Cref{def_QPref}). Then, there are several notions of strategic manipulation that might seem natural, including ideas such as stochastic dominance (SD). We define three different types of strategic manipulation that apply in this setting, and discuss why some others do not apply (including SD). We start by applying these definitions to QSWFs (\Cref{def_QSMW}), and then to QSCFs (\Cref{def_QSMC}). Then, we show that a QSWF that is QIC becomes QIC in the QSCF sense when composed with the NQSCE (\Cref{thm_QIC_QSWFs_are_QIC}) and that QCV is QIC (\Cref{thm_qcv_qic}). Therefore, QCVNE is QIC (\Cref{cor_QCVNE_is_QIC}).
    
        We then extend the definition of quantum dictatorship from \cite{sun2018quantum} to QSCFs, and show that this property is partially maintained over the NQSCE (\Cref{lem_non_sharp_dict_QSWF_still_non_sharp_dict}). Afterward, we show that in the case of QCVNE it is entirely maintained. We also define onto (\Cref{def_onto_qscf}), show that QCVNE satisfies this property, and arrive at our primary result, that QCVNE is a counterexample to \Cref{QGS}.

        \begin{theorem}[Disproof of Quantum Analogue of Gibbard-Satterthwaite]\label{thm_violate_QGS}
            The Quantum Gibbard-Satterthwaite Conjecture (\ref{QGS}) is false.
        \end{theorem}

    \textit{This work is at the intersection of several different fields, and uses terminology from each of them. For the readers convenience, Appendix \ref{app_ref_table} contains a table of many of the useful terms, symbols, and acronyms as a quick reference. Appendix \ref{app_defns} contains more formal definitions and theorems that are referenced throughout the paper. Depending on the reader's background, they may be already familiar with these. These definitions and theorems were omitted from the main text for the sake of flow. Many are either fairly intuitive or not strictly required for comprehension.}

\section{Preliminaries}
    To properly discuss attempting to violate a quantum analogue of Gibbard-Satterthwaite's Impossibility Theorem, a preliminary discussion on Classical Voting Systems, Social Choice and Welfare Functions, their properties, Arrow's Theorem, and the quantum analogues of each of the preceding topics is required.

    \subsection{Classical and Quantum Voting Systems}

        An understanding of classical voting systems is necessary to be able to properly discuss quantum voting systems. First is the definition of the participants of a voting system, and what they are voting on. Let $n$ be the number of voters, and $i$ be used to designate the $i$-th voter, denoted as $v_i$. Let $m$ be the number of alternatives. We define $A = \{a_1, \dots, a_m\}$ to be the set of alternatives and $a_j \in A$ as the $j$-th alternative. It is important to clarify how a voter votes, and the result of voting. We will discuss the first next, in the form of ballots.

        \subsubsection{Ballots}
             For any voting system, the type of ballot must be determined. For the purposes of this discussion, voting will be performed by a ballot formed from a single ranking among the possible alternatives, determining preference. This forms the input to our voting system. This also allows an easy extension to compare the utility of two alternatives for a specific voter, as it simply becomes proportional to the relative ranking between those alternatives by a voter (See Definition \ref{def_ranking}). These rankings for a set are denoted as $\mathcal{L}(A)$.

    
            In a quantum setting, a paradigm shift is required. Ballots are not a single ranking, but a quantum state of rankings, each forming a basis vector.  This allows for probabilistic and correlated voting (between voters and also alternatives).

            The authors acknowledge that many readers are unfamiliar with quantum objects, so we present Table \ref{tab_classical_vs_quantum} as a rough (mathematically loose) intuition for quantum objects in a voting system vs their classical counterparts.  Additionally, for intuition as to how a quantum voting system would function, see Figure \ref{fig_qcv_flow} where we explain Quantum Condorcet Voting.

            \begin{table}[ht]
            \centering
                \begin{tabular}{c|c|c|c|c}\label{tab_classical_vs_quantum}
                                & Classical & Classical Example & Quantum & Quantum Example \\ \hline
     Voters                     & $v_i$     & \{1, ..., $n$\}       & $v_i$   & \{1, ..., $n$\}   \\ \hline
     Alternatives               & $A$       & \{1, ..., $m$\}       & $A$     & \{1, ..., $m$\}   \\ \hline
     
     Voter ballots              & Ranking over $A$ & $1 \succ 3 = 4 \succ 2$ & \shortstack{Quantum state $\rho_i$ over\\all possible rankings\\over $A$} & \shortstack{$\rho_i = \frac{1}{\sqrt{2}}(\ket{1 \succ 3 = 4 \succ 2} $\\$ + \ket{3 = 2 \succ 4 \succ 1})$} \\ \hline
     
     Preference                 & \shortstack{``$v_i$ prefers $a$ \\over $b$" \\(\textit{sharp preference})} & $a \succ b$ & \shortstack{``$v_i$ prefers $a$ over $b$ with \\nonzero probability $p$" \\ ``$v_i$ has support for\\$a$ over $b$"} & $Tr\left(\Pi^{a \succ b}(\rho_i)\right) = p$ \\ \hline
     
     \shortstack{Social Welfare \\Function} & \shortstack{Maps voter\\ballots to\\final societal\\ranking} & \shortstack{Borda, Condorcet,\\Plurality, etc.} & \shortstack{Maps quantum voter\\ballots to final\\quantum state (which is\\ essentially a\\probability distribution\\over rankings)} & \shortstack{Quantum majority rule \\(unintuitive name since\\classical majority rule\\is not a welfare\\function) \cite{bao2017quantumarrows},\\Quantum Condorcet\cite{sun2021schrodinger}} \\ \hline
                 
                \end{tabular}~\newline
            \caption{Mathematically loose comparison of classical and quantum voting terminology.}
            \end{table}
            
            One important thing to note is that it is easy to interpret quantum voting as just probabilistic voting from this table. This is not technically correct. While society's ballot produced from a quantum social welfare function does indeed look exactly equal to the space of joint probability distributions, that fails to capture the essence of quantum voting. In particular, how voters are allowed to vote, how votes get combined, and how voters can easily create ``joint" distributions together (using the property of quantum entanglement). See \cite{bao2017quantumarrows} for an introductory discussion of quantum voting tactics and Appendix \ref{sec_intro_probvsquant} for our detailed discussion on probabilistic vs quantum voting. In short, calling quantum voting the same as probabilistic voting both minimizes the structure in place that allows voters to perform powerful actions or create very large joint probabilistic distributions in polynomial (or less) time and space and disregards inherently quantum inputs and outputs in the voting system.
            
            We now present the quantum voting formalisms, many of which are found in \cite{sun2021schrodinger}.
            
                \begin{definition}[Hilbert Space for Ordinal Rankings \cite{sun2021schrodinger}]\label{def_hilbert_rankings}
                    Let $\mathfrak{R}$ be a Hilbert space of dimension $|\mathcal{L}(A)|$. That is, $\mathfrak{R} = \mathbb{C}^{|\mathcal{L}(A)|}$, with each linear ranking $R \in \mathcal{L}(A)$ as a basis vector. We define $\mathfrak{R}_i$ to be an isomorphic Hilbert space to $\mathfrak{R}$ associated with voter $v_i$.
                \end{definition} 
                
                \begin{definition}[Hilbert Space of Alternatives]\label{def_hilbert_alternatives}
                    Let $\mathcal{A}$ be to Hilbert space of dimension $|A|$. That is, $\mathcal{A} = \mathbb{C}^{|A|}$, with each alternative $a \in A$ as a basis vector.
                \end{definition}
    
                \begin{definition}[Density Operators and Quantum Ballot Profiles \cite{sun2021schrodinger}]\label{def_density_op}
                    We define $D(H)$ as the set of density operators on the Hilbert Space $H$, and a \textit{Quantum Ballot} $\rho_i \in D(\mathfrak{R}_i)$ (specific voter's ballot) is a specific density operator on $\mathfrak{R}_i$.  Additionally, we define a \textit{Quantum Ballot Profile} $\rho \in D(\mathfrak{R}_1 \otimes \dots \otimes \mathfrak{R}_n)$ (collective ballot profile).
                \end{definition}

            This allows for submission of ballots as quantum states, and correspondingly quantum voting. This in turn allows movement beyond some of the limitations of classical ballot based systems. Performing operations on this space requires more definitions that allow us to reason about these quantum ballots.
            
                \begin{definition}[Projection \cite{sun2021schrodinger}]\label{def_projection}
                    Consider a pair $(x, y)$ of alternatives.  $\mathfrak{R}$ decomposes into subspaces associated with the possible relationships between $x$ and $y$.  To address these subspaces, we formally denote the subspace spanned by the elements of $\mathfrak{R}$ that encode $x \succ y$ (such as $|x \succ y \succ z \rangle, |z \succ x \succ y \rangle$) by $S^{x \succ y}$.  We denote $\Pi^{x \succ y}$ as the projection operator to the subspace $S^{x \succ y}$.
                \end{definition}
                \begin{definition}[Trace \cite{sun2021schrodinger}]\label{def_trace}
                    The trace of a space or subspace is defined as the probability of measuring a particular outcome on a particular instantiation of $\mathfrak{R}$.  Its use will look something like this: Tr$(\Pi^{x \succ y}(\text{Tr}_{\ne i}(\rho)))$ which reads (from innermost to outermost) ``looking at the quantum ballot profile $\rho$, looking specifically at the part of the ballot contributed to by voter $v_i$, then looking at how they ranked $x$ comparatively with $y$, and finally reading out the probability that they gave to the ranking $x$ over $y$."  We often use $\rho_i = \text{Tr}_{\ne i}(\rho)$.
                \end{definition}

            
            The projection operator allows us to study how two alternatives are comparatively ranked without being influenced by the rest of the elements. This is an inherent part of how QCV is formulated by Sun \textit{et al.} and when combined with trace allows these states to be reasoned within a simpler mathematical context.

            
        \subsubsection{Quantum Social Welfare Functions (QSWFs)}
            Quantum versions of the two primary function classes become necessary. In the literature, only the Quantum Social Welfare Function case has been well defined. This is as follows:
        
                \begin{definition}[Quantum Social Welfare Function (QSWF) \cite{sun2021schrodinger}]
                    A \textit{Quantum Social Welfare Function} is a function $\mathcal{E}: D(\mathfrak{R}_1 \otimes \dots \otimes \mathfrak{R}_n) \to D(\mathfrak{R})$, which can then be measured to determine a single ranking. We denote $\rho_{soc} = \mathcal{E}(\rho_1 \otimes \dots \otimes \rho_n)$ to be the resulting societal ballot.
                \end{definition}

            Quantum Social Choice Functions (QSCFs) are defined later (\Cref{def_QSCF}), as they have not been defined in the literature. Now, we can move to properties of each of these functions, and quantum counterparts where such counterparts have been established.

    \subsection{Properties of Classical and Quantum Voting Systems}
        As previously mentioned, voting systems have a set of axioms or properties that may or may not hold for each system. These allow for methods of comparison between them. There are many of these, and some relevant ones such as those present in Sun \textit{et al.} \cite{sun2021schrodinger} are in Appendix \ref{app_defns}. Those relevant to our goal are listed bellow.  The classical versions are pulled from various literature mentioned previously, and the quantum definitions are pulled from \cite{sun2021schrodinger}. Unanimity and IIA, which are not used in this work, are present in Appendix \ref{app_defns}.


        \subsubsection{Dictatorship} 
            The main property we will discuss that applies to SWFs, and the first to apply to SCFs, is Dictatorship. Classically, a dictatorship is defined as follows.
             
                \begin{definition}[Dictatorship for SWF]
                    An SWF $F$ is a \textit{Dictatorship} if there exists a $v_i$ such that \newline
                    $\forall (R_1, \dots, R_n) \in \mathcal{L}(A)^{n}$ where $F(R_1, \dots, R_i, \dots, R_n) = \prec$ then for all pairs of distinct candidates $(a_j, a_k)$,  $a_j \prec_i a_k \Rightarrow a_j \prec a_k.$
                \end{definition}
                
                \begin{definition}[Dictatorship for SCF]
                    An SCF $f$ is a \textit{Dictatorship} if there exists a $v_i$ such that \newline
                    $\forall (R_1, \dots, R_n) \in \mathcal{L}(A)^{n}$ where $f(R_1, \dots, R_i, \dots, R_n) = a_j,$ $\forall a_k, a_k \neq a_j \Rightarrow a_k \prec_i a_j.$
                \end{definition}

            This property is fairly intuitive, and the reasons for desiring a SWF or SCF to not have this property is as well. In a fair voting system, the final outcome should not be determined by a single voter. If it is, that single voter is the only person that needs to vote anyway.

            Such an idea has a quantum equivalent. Both Bao and Halpern \cite{bao2017quantumarrows} and Sun \textit{et al.} \cite{sun2021schrodinger} define this property as it applies to QSWFs.

                \begin{definition}[Quantum Dictatorship for QSWFs \cite{sun2021schrodinger}]~\newline
                    A QSWF $\mathcal{E}$ satisfies \textit{Sharp Dictatorship} if there is a voter $v_i$ such that for all quantum ballot profiles $\rho$  and all pairs of alternatives $(a, a')$:
                    \begin{itemize}
                        \item $\text{Tr}(\Pi^{a \succ a'}(\mathcal{E}(\rho))) \Leftrightarrow \text{Tr}(\Pi^{a \succ a'}\rho_i) = 1$
                    \end{itemize}
                    A QSWF $\mathcal{E}$ satisfies \textit{Unsharp Dictatorship} if there is a voter $v_i$ such that for all quantum ballot profiles $\rho$  and all pairs of alternatives $(a, a')$,
                    \begin{itemize}
                        \item $\text{Tr}(\Pi^{a \succ a'}(\mathcal{E}(\rho))) > 0 \Leftrightarrow \text{Tr}(\Pi^{a \succ a'}\rho_i) > 0$
                    \end{itemize}
                \end{definition}

            Otherwise stated, quantum sharp dictatorship occurs when the QSWF prefers alternative $x$ to alternative $y$ with certainty if only if the dictator does. Unsharp is defined similarly, but whenever the dictator prefers alternative $x$ to $y$ with nonzero probability.

            This leads us to the second property needed for Gibbard-Satterthwaite as it applies to SCFs.

        \subsubsection{Incentive Compatible} 
            The second and final property we will define on SCFs is the axiom of Incentive Compatibility. Classically, this is defined as:
             
                \begin{definition}[Incentive Compatible \cite{alggametheorybook}]
                    An SCF $f$ is \textit{Incentive Compatible} if it cannot be strategically manipulated. An SCF $f$ can be \textit{Strategically Manipulated} if $\exists v_i$ such that for some
                    $(R_1, \dots, R_n) \in \mathcal{L}(A)^{n}$ and some $R_i' \in \mathcal{L}(A)$ where $a_j \prec_i a_k$ but $f(R_1, \dots, R_i, \dots, R_n) = a_j$ and $f(R_1, \dots, R_i', \dots, R_n) = a_k$, we have that $ \Rightarrow a_j \prec a_k.$
                \end{definition}
             
            A natural interpretation of strategic voting is that a voter can ensure a preferred alternative can be chosen by strategically misrepresenting their vote. Incentive compatibility then requires that it is at least as good to report an accurate, truthful ranking from the voters perspective as a non-truthful one.

    \subsection{Arrow's Theorem and Violation in a Quantum Setting}
        These axioms allow us to compare different voting systems, but unfortunately this results in interactions between these properties that are not always desirable. A well known such result is Arrow's Impossibility Theorem (\Cref{thm_arrows}).
        
        This result tells us that no classical SWFs can have these desirable properties (IIA, Unanimity, and non-Dictatorship) while allowing for more than two alternatives (which is required for many interesting cases). An ability to violate this theorem is therefore desirable.
        
        Recent works have suggested that a quantum system of voting are a method of bypass. In two recent works \cite{bao2017quantumarrows, sun2021schrodinger}, quantum analogs of Arrow's Theorem have been defined, though slightly differently. In both works, their definition of such a quantum analogue is then shown to be bypassed (for one such definition, see Conjecture \ref{conj_q_arrows}). In the second, they use what they call Quantum Condorcet Voting.
        

    \subsection{Quantum Condorcet Voting (QCV)}
        Condorcet voting is a well known classical voting system based on the comparison of two alternatives across different ballots. This quantum voting system is related in some sense to that system, and is defined as follows. For a visual walkthrough of QCV, please see Figure \ref{fig_qcv_flow}.

        \begin{figure}
            \centering 
            \includegraphics[scale=0.5]{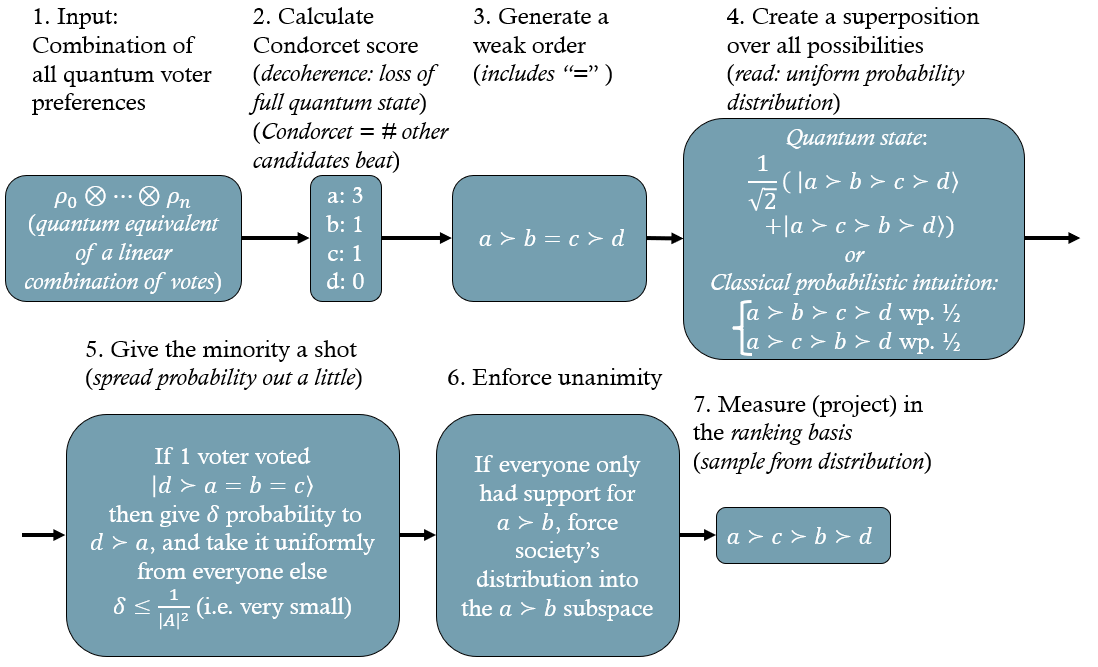}
            \caption{An intuition-based walk-through of QCV (mathematically loose).}
            \label{fig_qcv_flow}
        \end{figure}

            \begin{definition}[Quantum Condorcet Voting \cite{sun2021schrodinger}]
                Let $\rho_1 \otimes...\otimes \rho_n$ be a basis quantum ballot profile.  The quantum Condorcet voting QSWF $\mathcal{E}_{qcv}$ operates in the following steps:
    
                \begin{enumerate}
                    \item Calculate the \textit{Condorcet Score} of each alternative according to $\rho_1 \otimes ... \otimes \rho_n$.  The \textit{Condorcet Score} of an alternative is the number of times that alternative wins in all pairwise comparisons with other alternatives.  Formally, for an alternative $x$, their \textit{Condorcet Score} $S_c(x) = |\{y \in A : |\mathcal{V}^R_{x \succ y}| \ge |\mathcal{V}^R_{y \succ x}| \}|$ where R is the classical ballot profile corresponding to $\rho_1 \otimes ... \otimes \rho_n$.
                    \begin{itemize}
                        \item See Appendix \ref{app_defns} for formal definition of $\mathcal{V}$, which represents the set of voters.
                    \end{itemize}
    
                    \item Generate a weak order $\succeq$ over all alternatives according to their Condorcet score.  Formally, $x \succeq y$ iff $S_c(x) \ge S_c(y)$.
                    \begin{itemize}
                        \item This generates a ranking that could look something like $a \succ b = c \succ d$.
                    \end{itemize}
    
                    \item Complete the weak order.  Formally, generate the set $\{\succ^1, ..., \succ^w\}$ in which each $\succ^i$ is a linear order that extends $\succeq$, and $\{\succ^1, ..., \succ^w\}$ contains all extensions of $\succeq$.
                    \begin{itemize}
                        \item This transforms the ranking from step (2) to something that looks like this: $\{\succ^1, \succ^2\}$ where $\succ^1 = a \succ b \succ c \succ d$ and $\succ^2 = a \succ c \succ b \succ d$.
                    \end{itemize}
    
                    \item Transform the linear order into a quantum state.  Formally, for $\{\succ^1, ..., \succ^w\}$ we create a quantum state $\sigma^1 = \frac{1}{w} \sum_i \sigma_i$, where each $\sigma_i$ is the basis ballot that corresponds to $\succ^i$.
                    \begin{itemize}
                        \item Finally, we generate a quantum state out of this, creating a quantum superposition of (equally distributed probabilities across) all $\succ^i$.
                    \end{itemize}
    
                    \item \textit{Give the minority a shot}.  For any alternative pair $(x, y)$ which is encoded by at least one $\rho_i$ (individual voter's ballot), we ``spread" an amount $\delta \in (0, 1)$ of weight across the $x \succ y$ subspace.  Specifically, $\sigma^1$ is changed to $\sigma^2 = (1-k\delta)\sigma^1 + \delta\Omega^{x_1 \succ y_1} + ... + \delta\Omega^{x_k \succ y_k}$, where $(x_1, y_1), ..., (x_k, y_k)$ ranges over all alternative pairs that are encoded by at least one $\rho_i$.  The $\delta$ parameter is subject to $\delta < \frac{1}{|A|^2}$.
    
                    \item \textit{Enforce unanimity}. For any alternative pair $(x, y)$ that is encoded by all the $\rho_i$, we project $\sigma^2$ onto the $x \succ y$ subspace. Formally, $\sigma^2$ is transformed to $\sigma^3 = \frac{\Pi^{x_k \succ y_k}...\Pi^{x_1 \succ y_1} \sigma^2 \Pi^{x_1 \succ y_1}...\Pi^{x_k \succ y_k}}{\text{Tr}(\Pi^{x_k \succ y_k}...\Pi^{x_1 \succ y_1} \sigma^2)}$, where $(x_1, y_1), ..., (x_k, y_k)$ ranges over all alternative pairs that are encoded by all the $\rho_i$.
                \end{enumerate}
            \end{definition}

        Both \textit{giving the minority a shot} and \textit{enforcing unanimity} are first introduced by \cite{bao2017quantumarrows} and affirmed by \cite{sun2021schrodinger}.  They look strange at first, but are required and useful in proving that QCV disproves the quantum formulation of Arrow's Impossibility Theorem.

        This brings us to the result that inspires this work, and the main results of the previous works in this area \cite{bao2017quantumarrows, sun2021schrodinger}. Specifically, \cite{sun2021schrodinger} shows the following result.         
        
            \begin{theorem}[QCV Bypasses the Quantum analogue of Arrow's Theorem \cite{sun2021schrodinger}] \label{thm_violate_arrows}
                QCV fulfills the properties of sharp and unsharp unanimity, sharp and unsharp IIA, and is not a sharp or unsharp dictatorship, and as such bypasses the Quantum analogue of Arrow's Theorem.
            \end{theorem}

    \subsection{Gibbard-Satterthwaite}
        Our final definition is directly related to our goal, as it is the classical definition of GS. Similarly to Arrow's Impossibility Theorem, GS is a well known result on SCFs, sometimes as a corollary to Arrow's Theorem.
         
            \begin{theorem}[Gibbard-Satterthwaite \cite{alggametheorybook}]
                Let $f$ be an inventive compatible social choice function onto $A$, where $|A| > 2$, then $f$ is a dictatorship.
            \end{theorem}

        A quantum analogue of this does not yet exist in the literature, and as such we present our definition and justification of GS in a quantum setting.

\section{Results}
    In this work, we first define a way to emulate SCFs in the quantum setting. Second, we formulate a quantum analoguefor Incentive Compatibility. Third, we use this property along with a formulation of onto to formulate a quantum analogueof Gibbard-Satterthwaite's Impossibility Theorem. Finally, we show that a natural extension of QCV that forms an emulated SCF is incentive compatible, onto and maintains the property of non-dictatorship, and correspondingly violates this quantum analog.

    \subsection{Quantum Social Choice Functions (QSCFs)}
        First, we will define a Quantum analogue for SCFs in the same sense that QSWFs are defined.
         
            \begin{definition}[Quantum Social Choice Function (QSCF)]\label{def_QSCF}
                A \textit{Quantum Social Choice Function} is a function \newline $\xi: D(\mathfrak{R}_1 \otimes \dots \otimes \mathfrak{R}_n) \to D(\mathcal{A})$. That is, a function from the joint state formed by all voters' quantum ballots (density operators on their Hilbert Spaces) to density operators on the Hilbert Space of Alternatives. We denote $\alpha_{soc} = \mathcal{\xi}(\rho_1 \otimes \dots \otimes \rho_n)$ to be the resulting societal preference.
            \end{definition}

            \begin{definition}[Subspace for a Winner] \label{def_sub_winner}
                    A related useful definition is the subspace of a particular winner within the set of alternatives. Consider that a winner for a ranking must be ranked over every other alternative. To find if there is support for a candidate $a$ to be ranked over two candidates $b$ and $c$, you would take the intersection of the subspaces for each condition (that is, the condition of $a \succ b$ and $a \succ c$), which results in a subspace of $\mathfrak{R}$ where both conditions are true. Support on this subspace would then indicate support for the alternative $a$ being the winner of the measured ranking (though not necessarily 100\% support). This allows for more simple comparison of $D(\mathcal{A})$ and $D(\mathfrak{R})$.
                    \begin{equation*}
                        \mathcal{S}^{a} := \bigcap_{a' \neq a}\mathcal{S}^{a \succ a'}\\
                    \end{equation*}
                    We denote $\Pi^a$ to be the projection operator to this subspace.  It is important to note an interpretational difference in this context between $Tr(\Pi^a(\cdot))$ and $Tr(\Pi^{a \succ a'}(\cdot))$.  While $Tr(\Pi^a(\cdot))$ is interpreted as the probability of $a$ being the highest ranked alternative, $Tr(\Pi^{a \succ a'}(\cdot))$ is interpreted as the probability of $a$ being ranked over $a'$.
            \end{definition}

        An example of a QSCF $\xi$ would be selecting the ballot of the first voter (dictatorship) followed by a mapping of support on each basis vector (ranking) to support on that basis vector's highest ranked alternative. A use of this function might look like the following:
        \begin{align*}
            \rho_1 &= \frac{1}{\sqrt{2}} \left( \ket{x \succ y \succ z} + \ket{y \succ x \succ z} \right)\\
            \rho_2 &= \ket{z \succ x \succ y} \\
            \xi(\rho_1 \otimes \rho_2) &= \sum_{a \in A}Tr(\Pi^{a}\rho_1) \ket{a} = \frac{1}{\sqrt{2}} \left(\ket{x} + \ket{y}\right)
        \end{align*}

        \subsubsection{Quantum Social Choice Extensions}

            As a brief aside, we will discuss a particular way to form a Quantum Social Choice Function by utilizing an auxiliary function composed with a Quantum Social Welfare Function. This is, in fact, the method by which the previous example was formed, as will be seen below.
             
                \begin{definition}[Quantum Social Choice Extension (QSCE)]\label{def_QSCE}
                    A \textit{Quantum Social Choice Extension} is a function $h: D(\mathfrak{R}) \to D(\mathcal{A})$, which, when composed with a Quantum Social Welfare Function, creates a Quantum Social Choice Function.
                \end{definition}

                \begin{definition}[Natural Quantum Social Choice Extension (or the Natural Extension)]\label{def_NQSCE} 
                The most obvious of these extensions, which we will call the \textit{Natural Quantum Social Choice Extension} ($\Lambda$), would be a mapping of the support of each ranking in $\mathfrak{R}$ to equivalent support for its highest ranked alternative in $\mathcal{A}$. This is equivalently defined as the following
                    \begin{equation}\label{eq_NQSCE_dfn}
                        \Lambda(\rho) = \sum_{a \in A}Tr(\Pi^{a}\rho) \ket{a}
                    \end{equation}
                \end{definition}

            Using the previous example, the QSCF ($\xi$) could be written as a QSWF ($\mathcal{E}$, which in this case is a pure dictatorship of the first voter) and a QSCE ($\Lambda$):
            \begin{align*}
                \xi &= \Lambda \circ \mathcal{E}\\
                \rho &= \mathcal{E}(\rho_1 \otimes \rho_2) = \rho_1\\
                \Lambda(\rho_1) &= \sum_{a \in A}Tr(\Pi^{a}\rho_1) \ket{a} = \frac{1}{\sqrt{2}} \left(\ket{x} + \ket{y}\right)\\
                \xi(\rho_1 \otimes \rho_2) 
                    &= \Lambda(\mathcal{E}(\rho_1 \otimes \rho_2)) = \frac{1}{\sqrt{2}} \left(\ket{x} + \ket{y}\right)\\
            \end{align*}
            The use of QSCEs allow properties that have been proven about particular QSWFs to be quickly adapted to form QSCFs. These then can be shown to satisfy related properties with minimal additional effort, such as the property we define next. 
    
    \subsection{Quantum Incentive Compatibility}
        To work with GS in a quantum setting, we also need some notion of incentive compatibility. We came to the following formulation, in the spirit of the quantum formulations of other such properties, such as in \cite{bao2017quantumarrows, sun2021schrodinger}. The intuition for incentive compatibility is that a SWF does not rank alternatives in such a way that voters have an incentive to lie to try to manipulate the vote, which is called strategic manipulation. Thus, we define Quantum Incentive Compatibility as follows.

            \begin{definition}[Quantum Incentive Compatibility (QIC)]
                A QSWF or QSCF is \textit{Quantum Incentive Compatible} if it is not subject to Quantum Strategic Manipulation.
            \end{definition}

        This is, metaphorically speaking, just kicking the can down the road, as we require a definition of Quantum Strategic Manipulation. There are two parts to this, similar to cause and effect. The first is define what a preference is in the quantum setting, which is the ``cause". The second is what it means to manipulate a quantum social choice, which is the ``effect", modelled in the spirit of the other previously defined properties.
        
        \subsubsection{Quantum Preferences}

            To discuss what a voter's goal would be within voting, we need to define the notion of a preference. Classically, a preference is clear by the voter's ballot, as a single total ranking is usually given, which gives the ordinal ranking of utility for the voter of each of the alternatives. This does not strictly apply to the quantum setting due to reasons expressed in Appendix \ref{sec_intro_probvsquant} and later in \Cref{sec_StMa_discussion}, so we define a notion of quantum preference.
             
            \begin{definition}[Quantum Preference] \label{def_QPref}
                A voter $v_i$ has a \textit{Strong Positive Quantum Preference} for alternative $a$ over $a'$ if, with truthful ballot $\rho_i$, $\text{Tr}(\Pi^{a \succ a'}(\rho_i)) = 1$. That is, the voter's ballot $\rho_i$ only has support on the subspace for that preference ($S^{a \succ a'}$). 
                Similarly, that voter has a \textit{Strong Negative Quantum Preference} against alternative $a$ over $a'$ if $\text{Tr}(\Pi^{a \succ a'}(\rho_i)) = 0$. That is, $\rho_i$ has no support on the subspace for that preference ($S^{a \succ a'}$).
                Finally, that voter has a \textit{Weak Quantum Preference} if $\text{Tr}(\Pi^{a \succ a'}(\rho_i)) > 0$. That is, if $\rho_i$ has any support on the subspace for that preference ($S^{a \succ a'}$).
            \end{definition}

            In more easily understood language,
            \begin{itemize}
                \item A \textit{Strong Positive Quantum Preference} means a voter gives probability 1 to ranking $x$ over $y$ (e.g. $\frac{1}{\sqrt{2}} (\ket{x \succ y \succ z} + \ket{z \succ x \succ y})$).
                \item A \textit{Strong Negative Quantum Preference} means a voter gives probability 0 to ranking $x$ over $y$ (e.g. $\frac{1}{\sqrt{3}} (\ket{y \succ z \succ x} + \ket{z \succ y \succ x} + \ket{y \succ x \succ z})$).
                \item A \textit{Weak Quantum Preference} means a voter gives non-zero probability to ranking $x$ over $y$ \newline$\left(\text{e.g. } \sqrt{\frac{1}{3}}\ket{x \succ y \succ z} + \sqrt{\frac{2}{3}}\ket{y \succ z \succ x}\right)$.
            \end{itemize}

        \subsubsection{Quantum Strategic Manipulation}
            This leads us to the definition of what it means to manipulate the outcome, which is the ``effect". We also need to ensure that these definitions apply to all QSWFs and QSCFs in general, not just the particular functions we use in our result. With this in mind, we define strategic manipulation as follows.

                \begin{definition}[Quantum Strategic Manipulation (Welfare)]\label{def_QSMW}
                     Without loss of generality, let there exist a voter $v_i$ with truthful quantum ballot $\rho_i \in D(\mathfrak{R}_i)$, and correspondingly truthful ballots for all other voters. Given a QSWF $\mathcal{E}$, let $\rho_{soc} = \mathcal{E}(\rho_1 \otimes \dots \otimes \rho_i \otimes \dots \otimes \rho_n)$. Again, without loss of generality $\mathcal{E}$ is subject to \textit{Quantum Strategic Manipulation} with respect to alternatives $a$ and $a'$ if there exists an alternate quantum ballot $\rho_i' \in D(\mathfrak{R}_i)$ and $\rho_{soc}' = \mathcal{E}(\rho_1 \otimes \dots \otimes \rho_i' \otimes \dots \otimes \rho_n)$ such that $v_i$ has any of the following preferences and can manipulate the societal outcome in the corresponding way:
                     \begin{itemize}
                         \item Given a Strong Positive Quantum Preference for $a$ over $a'$ ($\text{Tr}(\Pi^{a \succ a'}(\rho_i)) = 1$), the voter can change society's ballot to only have support for this preference (that is, $\text{Tr}(\Pi^{a \succ a'}(\rho_{soc})) < 1$, but $\text{Tr}(\Pi^{a \succ a'}(\rho_{soc}')) = 1$).
                         \item Given a Strong Negative Quantum Preference against $a$ over $a'$ ($\text{Tr}(\Pi^{a \succ a'}(\rho_i)) = 0$), the voter can change society's ballot to have no support for this preference (that is, $\text{Tr}(\Pi^{a \succ a'}(\rho_{soc})) > 0$, but $\text{Tr}(\Pi^{a \succ a'}(\rho_{soc}')) = 0$).
                         \item Given a Weak Quantum Preference for $a$ over $a'$ ($\text{Tr}(\Pi^{a \succ a'}(\rho_i)) > 0$), the voter can change society's ballot to have support for this preference (that is, $\text{Tr}(\Pi^{a \succ a'}(\rho_{soc})) = 0$, but $\text{Tr}(\Pi^{a \succ a'}(\rho_{soc}')) > 0$).
                     \end{itemize}
                \end{definition}

            We see two more potential cases of strategic manipulation, though they are not included in our definition.
            Those are $\text{Tr}(\Pi^{a \succ a'}(\rho_i)) > \text{Tr}(\Pi^{a \succ a'}(\rho_i'))$ and $\text{Tr}(\Pi^{a \succ a'}(\rho_i')) > \text{Tr}(\Pi^{a \succ a'}(\rho_i))$ for negative and positive preferences respectively.
            These are not included because they do not seem to match the basic intuitions from the properties defined in other quantum works \cite{bao2017quantumarrows, sun2021schrodinger}. For example, QIIA is not defined based on the final probabilities being the same, but merely greater than zero or equal to 1. However, we wish to explore such notions for voting systems such as QCV in a future work.

            \paragraph{Other Notions of Strategic Manipulation.} \label{sec_StMa_discussion} Extending this line of thought, it might appear to be reasonable to use a notion of stochastic dominance (SD) (or even lex-truthfulness \cite{Chakrabarty2014welfare}) as is seen with probabilistic voting \cite{lederer2021non}. As noted in Appendix \ref{sec_intro_probvsquant}, these notions of strategy-proofness (SP-ness) do not make much sense in a quantum setting due to violation of the assumptions that these notions make. Specifically, both notions presume there to be strong or weak ordering preferences over outcomes. However, in a quantum setting, voters are allowed to maintain probability mixtures over alternatives as preferences, violating any classical sense of a strong or weak ordering. As mentioned in Appendix \ref{sec_intro_probvsquant}, PC-SP-ness, as discussed by Brandt \cite{brandt2017rolling}, is a notion of SP-ness which, due to its looser definition, might be more apt to describe quantum preferences. It simply states that one outcome is preferred if it is more likely to produce a more favorable outcome. However, this notion also fails to shed light on what ``favorable outcome" would look like when voter preferences are probabilistic. This essence is what we try to capture in our definition of quantum strategic manipulation.
            
            There is also an argument for attempting to match the voter's truthful ballot, but this reasoning allows for an unfortunate side effect. In the case of correlation between voters (such as quantum entanglement), this likely loses meaning as the final ballot is stand alone and does not have other ballots to be entangled with. Additionally, what it means to be ``closer" in a quantum context is not readily apparent.  This kind of manipulation is potentially very nuanced in the general probabilistic setting and may not to apply to the general quantum case. We do not include it in this work, though it may be interesting to study in the future.

            \paragraph{Positive Response.} Procaccia \cite{Procaccia_2010} implied a different perspective on the issue of strategic manipulation. He introduced the notion of \textit{positive response} as a possible desirable axiom that randomized rules may be required to satisfy. Positive response is informally defined as follows: ``\textit{for every agent there is some profile where the agent can increase an alternative's probability of being elected by pushing it upwards in its vote}'' \cite{Procaccia_2010}. If an agent can increase an alternative's probability of being elected in the societal vote, that implies that the agent is also \textit{decreasing} a different alternative's probability of being elected. Thus, it may become desirable for an agent to misrepresent their true preference probability mixture in favor of one that boosts a particular candidate or more closely aligns the final output with their own preference. We offer two examples. In the first, an agent $a_i$'s true preference is 80\% for alternative $x$ and 20\% for alternative $y$. If society outputs 100\% $x$, $a_i$ may want to report a higher percentage for alternative $y$ in order to more closely align society's preference with its own (receiving positive response for alternative $y$ by definition implies negative response for alternative $x$). In a second example, an agent $a_i$ has true preference \{$x$: 60\%, $y$: 30\%, $z$: 10\%\} and society's preference is \{$x$: 10\%, $y$: 50\%, $z$: 40\%\}. In this instance, it would be desirable for $a_i$ to try and elicit positive response on $x$ such that it boosts its highest preference. In doing so, this would inadvertently also bring either $y$ or $z$ or both closer to $a_i$'s true preference. Thus, the idea of \textit{positive response} implies the ability to do what is traditionally thought of as strategic manipulation. Not to be construed, by this discussion we mean to say that what is thought of as traditional strategic manipulation is significantly more nuanced in a probabilistic (or quantum) setting and is perhaps not the best understanding of strategic manipulation.

            \paragraph{Quantum Strategic Manipulation (Choice).} We must now discuss applying this definition to Quantum Social Choice Functions instead of Quantum Social Welfare Functions. The definitions are very similar to the previous cases.

                \begin{definition}[Quantum Strategic Manipulation (Choice)]\label{def_QSMC}
                     Without loss of generality, let there exist a voter $v_i$ with truthful quantum ballot $\rho_i \in D(\mathfrak{R}_i)$, and correspondingly truthful ballots for all other voters. Given a QSCF $\xi$, let $\alpha= \xi(\rho_1 \otimes \dots \otimes \rho_i \otimes \dots \otimes \rho_n)$. Again, without loss of generality $\xi$ is subject to \textit{Quantum Strategic Manipulation} with respect to alternative $a$ if there exists an alternative quantum ballot $\rho_i' \in D(\mathfrak{R}_i)$ and $\alpha' = \xi(\rho_1 \otimes \dots \otimes \rho_i' \otimes \dots \otimes \rho_n)$ such that $v_i$ has any of the following preferences and can manipulate the societal outcome in the corresponding way:
                     \begin{itemize}
                         \item Given a Strong Positive Quantum Preference for $a$ over all other candidates ($\text{Tr}(\Pi^{a}(\rho_i)) = 1$), the voter can change society's ballot to only have support for this alternative (that is, $\text{Tr}(\Pi^{a}(\alpha)) < 1$, but $\text{Tr}(\Pi^{a}(\alpha')) = 1$).
                         \item Given a Strong Negative Quantum Preference against $a$ ($\text{Tr}(\Pi^{a}(\rho_i)) = 0$, implying that $a$ cannot be the highest ranked candidate in any ranking $v_i$ gives support for in their ballot), the voter cannot remove all support for this alternative in society's ballot (that is, $\text{Tr}(\Pi^{a}(\alpha)) > 0$, but $\text{Tr}(\Pi^{a}(\alpha')) = 0$).
                         \item Given a Weak Quantum Preference for $a$ over all other candidates ($\text{Tr}(\Pi^{a}(\rho_i)) > 0$), the voter can change society's ballot to have support for this alternative (that is, $\text{Tr}(\Pi^{a}(\alpha)) = 0$, but $\text{Tr}(\Pi^{a}(\alpha')) > 0$).
                     \end{itemize}
                \end{definition}
             
            We note that our definition of Strong Negative Quantum Preference for welfare functions is intuitively stronger due to the enforcement of pairwise dominance, while for choice functions it only enforces an alternative to not be in the top position. We use this intuition to show that if a QSCF is built from the Natural QSCE and a QSWF that is not strategically manipulable, then it is also not strategically manipulable.
             
            \begin{theorem}[QIC QSWFs with the Natural Extension remain QIC] \label{thm_QIC_QSWFs_are_QIC}
                Any Quantum Social Choice Function $\xi$ formed by composing a Quantum Incentive Compatible Quantum Social Welfare Function $\mathcal{E}$ with the Natural Quantum Social Choice Extension $\Lambda$ is Quantum Incentive Compatible.
            \end{theorem}
             
            \begin{proof}
                We cover all three cases below:
                \begin{itemize}
                    \item (\textit{Strong Positive Quantum Preferences})
                    
                    Assume, for the sake of contradiction, that voter $v_i$ can manipulate $\xi$ with respect to Strong Positive Quantum Preferences for some alternative $a$ over all other $a'$ in their ballot $\rho_i$, and societal preference is formed as $\rho_{soc} = \mathcal{E}(\rho_0 \otimes \dots \otimes \rho_i \otimes \dots \otimes \rho_n)$ (QSWF) and then transformed into $\alpha = \Lambda(\rho_{soc})$ (QSCE), but that $\mathcal{E}$ is QIC. 
                    
                    To be able to strategically manipulate $\xi$, the voter must be able to create a $\rho_i'$, and correspondingly a societal preference $\rho_{soc}' = \mathcal{E}(\rho_0 \otimes \dots \otimes \rho_i' \otimes \dots \otimes \rho_n)$ that is transformed to $\alpha' = \Lambda(\rho_{soc}')$ such that $\text{Tr}(\Pi^{a}(\alpha)) < 1$ (by assumption), but $\text{Tr}(\Pi^{a}(\alpha')) = 1$. The second of those implies, by \Cref{eq_NQSCE_dfn} in \Cref{def_NQSCE}, $a$ is the top choice in all rankings with support in $\rho_{soc}'$, so $\forall a' \neq a, \text{Tr}(\Pi^{a \succ a'}(\rho_{soc}')) = 1$.
                    
                    By assumption, $\mathcal{E}$ is QIC and thus cannot be manipulated with respect to any Strong Positive Quantum Preference $a \succ a'$. Therefore $\forall a' \neq a, \text{Tr}(\Pi^{a \succ a'}(\rho_{soc}')) = 1$ implies $\forall a' \neq a, \text{Tr}(\Pi^{a \succ a'}(\rho_{soc})) = 1$ and, by \Cref{eq_NQSCE_dfn}, this implies $\text{Tr}(\Pi^{a}(\alpha)) = 1$. This contradicts $\text{Tr}(\Pi^{a}(\alpha)) < 1$. Therefore, $\xi$ cannot be strategically manipulated with respect to a Strong Positive Quantum Preference.

                    \item (\textit{Strong Negative Quantum Preference}) 

                    Assume, for the sake of contradiction, that voter $v_i$ can manipulate $\xi$ with respect to a Strong Negative Quantum Preference against some alternative $a$ (and therefore is never the top ranked candidate in any ranking they have support for) in their ballot $\rho_i$, and societal preference is formed as $\rho_{soc} = \mathcal{E}(\rho_0 \otimes \dots \otimes \rho_i \otimes \dots \otimes \rho_n)$ (QSWF) and then transformed into $\alpha = \Lambda(\rho_{soc})$ (QSCE), but that $\mathcal{E}$ is QIC. 
                    
                    To be able to strategically manipulate $\xi$, the voter must be able to create a $\rho_i'$, and correspondingly a societal preference $\rho_{soc}' = \mathcal{E}(\rho_0 \otimes \dots \otimes \rho_i' \otimes \dots \otimes \rho_n)$ that is transformed to $\alpha' = \Lambda(\rho_{soc}')$, such that $\text{Tr}(\Pi^{a}(\alpha)) > 0$ (by assumption), but $\text{Tr}(\Pi^{a}(\alpha')) = 0$. The second of those implies, by \Cref{eq_NQSCE_dfn} in \Cref{def_NQSCE}, that $a$ is never the top choice in any of the rankings with support in $\rho_{soc}'$, so $\exists a' \neq a ~s.t.~ \text{Tr}(\Pi^{a \prec a'}(\rho_{soc}')) = 1$.
                    
                    By assumption, $\mathcal{E}$ is QIC and thus cannot be manipulated with respect to any Strong Positive Quantum Preference $a \prec a'$. Therefore  $\text{Tr}(\Pi^{a \prec a'}(\rho_{soc}')) = 1$ implies $\text{Tr}(\Pi^{a \prec a'}(\rho_{soc})) = 1$ and $\text{Tr}(\Pi^{a \succ a'}(\rho_{soc})) = 0$. By \Cref{eq_NQSCE_dfn}, this implies $\text{Tr}(\Pi^{a}(\alpha)) = 0$. This contradicts $\text{Tr}(\Pi^{a}(\alpha)) > 0$. Therefore, $\xi$ cannot be strategically manipulated with respect to a Strong Negative Quantum Preference.
                     
                    \item (\textit{Weak Quantum Preference})

                    Assume, for the sake of contradiction, that voter $v_i$ can manipulate $\xi$ with respect to Weak Quantum Preferences for some alternative $a$ over all other $a'$ in their ballot $\rho_i$, and societal preference is formed as $\rho_{soc} = \mathcal{E}(\rho_0 \otimes \dots \otimes \rho_i \otimes \dots \otimes \rho_n)$ (QSWF) and then transformed into $\alpha = \Lambda(\rho_{soc})$ (QSCE), but that $\mathcal{E}$ is QIC. 
                    
                    To be able to strategically manipulate $\xi$, the voter must be able to create a $\rho_i'$, and correspondingly a societal preference $\rho_{soc}' = \mathcal{E}(\rho_0 \otimes \dots \otimes \rho_i' \otimes \dots \otimes \rho_n)$ that is transformed to $\alpha' = \Lambda(\rho_{soc}')$, such that $\text{Tr}(\Pi^{a}(\alpha)) = 0$ (by assumption), but $\text{Tr}(\Pi^{a}(\alpha')) > 0$. The second of those implies, by \Cref{eq_NQSCE_dfn} in \Cref{def_NQSCE}, that $a$ is the top choice in some ranking with support in $\rho_{soc}'$. That ranking's support indicates support for $a$ being preferred over every other candidate, so $\forall a' \neq a, \text{Tr}(\Pi^{a \succ a'}(\rho_{soc}')) > 0$.
                    
                    By assumption, $\mathcal{E}$ is QIC and thus cannot be manipulated with respect to any Weak Quantum Preference $a \succ a'$. Therefore $\forall a' \neq a, \text{Tr}(\Pi^{a \succ a'}(\rho_{soc}')) >0$ implies $\forall a' \neq a, \text{Tr}(\Pi^{a \succ a'}(\rho_{soc})) > 0$. By \Cref{eq_NQSCE_dfn}, this implies $\text{Tr}(\Pi^{a}(\alpha)) > 0$. This contradicts $\text{Tr}(\Pi^{a}(\alpha)) = 0$. Therefore, $\xi$ cannot be strategically manipulated with respect to a Weak Quantum Preference.
                \end{itemize}
            \end{proof}

            This implies that a QSCF being formed from a QIC QSWF composed with the NQSCE is an inherently stronger condition than a QSCF  being QIC (regardless of how it was formed), which matches the intuition mentioned before. With this in mind, it suffices to show that QCV is QIC to prove that one of the desired properties for bypassing QGS holds, leaving the properties of onto and non-dictatorship.

            \begin{theorem}[QCV is QIC] \label{thm_qcv_qic}
                Quantum Condorcet Voting is Quantum Incentive Compatible.
            \end{theorem}    
         
            \begin{proof}
                We will address each part of QIC individually.  If QCV respects each part of QIC, then it is QIC.
                \begin{itemize}
                    \item (\textit{Strong Positive Quantum Preference}) Given that voter $v_i$ has a Strong Positive Quantum Preference for $a$ over $a'$ ($\text{Tr}(\Pi^{a \succ a'}(\rho_i)) = 1$), they cannot increase society's support for $a \succ a'$ to 1 if it wasn't there already. That is, for some honest societal ballot $\rho_{soc}$ and some societal ballot $\rho_{soc}'$ where voter $v_i$ lied to try to boost the support for $a \succ a'$, $\text{Tr}(\Pi^{a \succ a'}(\rho_{soc})) < 1$ implies that $\text{Tr}(\Pi^{a \succ a'}(\rho_{soc}')) < 1$ \textit{when we use QCV}.
    
                    If $\text{Tr}(\Pi^{a \succ a'}(\rho_{soc})) < 1$, by the enforcing unanimity step in QCV (step 6), some other voter $v_j$ must have support on a subspace $S^{a \preceq a'}$. Since $v_j$ did not change their vote in $\rho_{soc}'$, by giving the minority a shot, society must have support on a subspace $S^{a \preceq a'}$, and therefore $\text{Tr}(\Pi^{a \preceq a'}(\rho_{soc}')) > 0$. Thus, because $\text{Tr}(\Pi^{a \preceq a'}(\rho_{soc}'))  + \text{Tr}(\Pi^{a \succ a'}(\rho_{soc}')) = 1$, $\text{Tr}(\Pi^{a \succ a'}(\rho_{soc}')) < 1$.
    
                    \item (\textit{Strong Negative Quantum Preference}) Given that voter $v_i$ has a Strong Negative Quantum Preference against $a$ over $a'$ ($\text{Tr}(\Pi^{a \succ a'}(\rho_i)) = 0$), they cannot decrease society's support for $a \succ a'$ to 0 if it wasn't there already. That is, for some honest societal ballot $\rho_{soc}$ and some societal ballot $\rho_{soc}'$ where voter $v_i$ lied to try to decrease the support for $a \succ a'$, $\text{Tr}(\Pi^{a \succ a'}(\rho_{soc})) > 0$ implies that $\text{Tr}(\Pi^{a \succ a'}(\rho_{soc}')) > 0$.
    
                    If $\text{Tr}(\Pi^{a \succ a'}(\rho_i)) = 0$ while $\text{Tr}(\Pi^{a \succ a'}(\rho_{soc})) > 0$, then by enforcing unanimity in QCV (step 6), some other voter $v_j$ must have support on the subspace $S^{a \succ a'}$. Since $v_j$ did not change their vote in $\rho_{soc}'$, by giving the minority a shot, society must still have support on the subspace $S^{a \succ a'}$, and thus $\text{Tr}(\Pi^{a \succ a'}(\rho_{soc}')) > 0$.
    
                    \item (\textit{Weak Quantum Preference}) Given that voter $v_i$ has a Weak Quantum Preference for $a$ over $a'$ ($\text{Tr}(\Pi^{a \succ a'}(\rho_i)) > 0$), they cannot change society's support for $a \succ a'$ to greater than 0 if it wasn't there already. That is, for some honest societal ballot $\rho_{soc}$ and and some societal ballot $\rho_{soc}'$ where voter $v_i$ lied to try to boost the support for $a \succ a'$, $\text{Tr}(\Pi^{a \succ a'}(\rho_{soc})) = 0$ implies that $\text{Tr}(\Pi^{a \succ a'}(\rho_{soc}')) = 0$.
    
                    This case is trivial. If a voter $v_i$ had support on $a \succ a'$, then society \textit{must} have support on $a \succ a'$ by giving the minority a shot. Thus, the presumptions are false and this case never applies.
                \end{itemize}
            \end{proof}
         
            \begin{corollary}[QCV with the Natural Extension is QIC] \label{cor_QCVNE_is_QIC}
                The QSCF formed from QCV with the Natural Extension (QCVNE) is QIC.
            \end{corollary}
         
            \begin{proof}
                Follows from Theorems \ref{thm_qcv_qic} and \ref{thm_QIC_QSWFs_are_QIC}.
            \end{proof}

    \subsection{Dictatorship for Quantum Social Choice Functions}

        Next, we will adapt the intuition of quantum dictatorship as defined in previous works \cite{bao2017quantumarrows, sun2021schrodinger} to the case of Quantum Social Choice Functions. The intuition is straight-forward in that it replaces full rankings with highest ranked candidates.

            \begin{definition}[Quantum Dictatorship for QSCFs]\label{def_dict_qscf}~\newline
                A QSCF $\xi$ satisfies \textit{Sharp Dictatorship} if there is a voter $v_i$ with ballot $\rho_i$ with $\alpha = \xi(\rho_1 \otimes \dots \otimes \rho_i \otimes \dots \otimes \rho_n)$, 
                \begin{itemize}
                    \item $\forall a, (\text{Tr}(\Pi^{a}(\alpha)) = 1 \Leftrightarrow \text{Tr}(\Pi^{a}(\rho_i)) = 1)$.
                \end{itemize}
                A QSCF $\xi$ satisfies \textit{Unsharp Dictatorship} if there is a voter $v_i$ with ballot $\rho_i$ with $\alpha = \xi(\rho_1 \otimes \dots \otimes \rho_i \otimes \dots \otimes \rho_n)$, 
                \begin{itemize}
                    \item $\forall a, (\text{Tr}(\Pi^{a}(\alpha)) >0 \Leftrightarrow \text{Tr}(\Pi^{a}(\rho_i)) >0)$.
                \end{itemize}
                A QSCF is a \textit{Quantum Dictatorship} if it is either a Sharp or Unsharp Dictatorship. 
            \end{definition}

        As with incentive compatibility, there appears to be a connection between dictatorial properties of QSWF and QSCF. Such a connection for the Sharp Dictatorship case is proved below. If such a connection exists for weak dictatorships as well is an open question. As we show later in the proof of \Cref{thm_violate_QGS} (\Cref{thm_violate_QGS_proof}), both cases hold for QCV, but they may or may not hold in the general case.

        \begin{lemma}[Non-Sharp Dictatorial QSWFs with the Natural Extension are non-Sharp-Dictatorial]\label{lem_non_sharp_dict_QSWF_still_non_sharp_dict} 
            Any Quantum Social Choice Function $\xi$ formed by composing a Quantum Social Welfare Function $\mathcal{E}$ that is not a Sharp Dictatorship with the Natural Extension $\Lambda$ is not a Sharp Dictatorship.
        \end{lemma}
         
        \begin{proof} 
            Assume, for the sake of contradiction, that voter $v_i$ is a Sharp Dictator for $\xi$  but not a Sharp Dictator for $\mathcal{E}$ with quantum ballot $\rho_i$ and quantum ballot profile $\rho$. Societal preference is formed as $\rho_{soc} = \mathcal{E}(\rho)$ (QSWF) and then transformed into $\alpha = \Lambda(\rho_{soc})$ (QSCE).
                    
            By Definition \ref{def_sub_winner}, an alternative has complete support for being the highest ranked candidate in a quantum ballot if and only if it has complete support on being ranked above each other alternative.

            Therefore, our first observation is that for all alternatives $a$:
            \begin{equation}\label{obs_choice_to_welfare_pi}
                \text{Tr}(\Pi^{a}(\rho_i)) = 1 \Leftrightarrow \forall a' \neq a, \text{Tr}(\Pi^{a \succ a'}(\rho_i)) = 1.
            \end{equation}
            
            This same reasoning, when combined with Definition \ref{def_NQSCE} (for the Natural Extension $\Lambda$), allows for a second observation, that for all alternatives $a$
            \begin{equation}\label{obs_choice_to_welfare_fctn}
                \text{Tr}(\Pi^{a}(\alpha)) = 1 \Leftrightarrow \forall a' \neq a, \text{Tr}(\Pi^{a \succ a'}(\rho_{soc})) = 1.
            \end{equation}

            Since $\xi$ is a Sharp Dictatorship, by Definition \ref{def_dict_qscf} and the previous two statements for all alternatives $a$:
            \begin{align}
                \forall a' \neq a, \text{Tr}(\Pi^{a \succ a'}(\rho_i)) = 1 
                &\Leftrightarrow \text{Tr}(\Pi^{a}(\rho_i)) = 1 \label{obs_sharp_dict_app_1} \\
                &\Leftrightarrow\text{Tr}(\Pi^{a}(\alpha)) = 1 \label{obs_sharp_dict_app_2} \\ 
                & \Leftrightarrow \forall a' \neq a, \text{Tr}(\Pi^{a \succ a'}(\rho_{soc})) = 1 \label{obs_sharp_dict_app_3} 
            \end{align}
            where (\ref{obs_sharp_dict_app_1}) is by Observation \ref{obs_choice_to_welfare_pi}, (\ref{obs_sharp_dict_app_2}) is by Definition \ref{def_dict_qscf}, and (\ref{obs_sharp_dict_app_3}) is by Observation \ref{obs_choice_to_welfare_fctn}.
            
            This matches the definition of Sharp Dictatorship for QSWFs in \Cref{def_quantum_dict}, and therefore implies $\mathcal{E}$ is a Sharp Dictatorship. This is a contradiction because $\mathcal{E}$ is not a dictatorship by assumption. Therefore $\xi$ cannot be a Sharp Dictatorship.
        \end{proof}

        Now that almost all of the required properties have been defined, we can move to our main result.
        
    \subsection{Quantum Analogue of Gibbard-Satterthwaite}

        One last definition is the definition of onto. Classically, onto refers to a function whose image encompasses the entire co-domain (that is, the entire output space has at least one member of the domain mapped to it). This definition of onto could technically apply here, but does not match our interpretation of the original purpose of onto for GS. The mathematical definition in this case would require that every possible density operator over $\mathcal{A}$ be a possible output from the QSCF (of which there are uncountably many), which we find does not match the intuition and does not serve any additional purpose in this case. The general intuition for our definition, which we define below, is as follows: for every candidate, there is some societal ballot where the result of the QSCF has support for only that candidate.
        
        The interpretation given in \cite{alggametheorybook}, which is similar to our interpretation, states that onto ensures that the condition of $|A| > 2$ ``has bite". That is, without it a voting system could ignore all but two candidates and bypass the result. Mathematical onto seems far too strict to match this intuition more than the one we define below. Regardless, many voting systems (including QCV) might require an infinite amount of voters to form all of the density operators due to their continuous nature. In QCV specifically, this is due to the first step: calculating the Condorcet score for each alternative for each voter. In this step, QCV discretizes the input, disregarding the significantly more complex density operators that were input.
         
            \begin{definition}[A Notion of Onto for QSCFs]\label{def_onto_qscf}
                A Quantum Social Choice Function $\xi$ is onto if $\forall a \in A, \exists \rho ~s.t.~ \xi(\rho) = \ket{a}$.
            \end{definition}

        Now that the last required property has been defined, this brings us to the goal of this work, the violation of a quantum analogue of the Gibbard-Satterthwaite Impossibility Theorem (\Cref{QGS}).

         
        \begin{proof}[Proof of \Cref{thm_violate_QGS}]\label{thm_violate_QGS_proof}
            We disprove the conjecture by counterexample.

            QCV with the Natural Extension is QIC by \Cref{cor_QCVNE_is_QIC}. It is not a sharp dictatorship as follows by \Cref{lem_non_sharp_dict_QSWF_still_non_sharp_dict} and \Cref{thm_violate_arrows} (which states that QCV violates Arrow's, including a proof that it is non-dictatorial in the QSWF sense). Thus, it suffices to show that QCV with the Natural Extension is onto and not a weak dictatorship.

            QCV with the Natural Extension is clearly onto (as defined in \Cref{def_onto_qscf}) as for any alternative $a$, if every voter votes with identical ballots $\rho_i = \ket{a \succ \dots}$ with $a$ as the highest ranked candidate, QCV will output $\rho_{soc} = \rho_i = \ket{a \succ \dots}$ by enforcing unanimity (step 6 of QCV) and $\Lambda(\rho_{soc}) = \ket{a}$ by \Cref{def_NQSCE}.

            Finally, QCV with the Natural Extension is not a Weak Dictatorship as any voter can submit a ballot with support for a ranking that has some alternative $a$ as the highest ranked alternative. That ranking must then have support in the societal quantum ballot $\rho_{soc}$ formed by QCV due to giving the minority a shot (step 5 of QCV), and $\alpha = \Lambda(\rho_{soc})$ will have support for the alternative $a$ by \Cref{def_NQSCE} regardless of any other voters ballot. This violates the if and only if requirement.

            Thus, Quantum Condorcet Voting with the Natural Quantum Social Choice Extension is a Quantum Incentive Compatible Quantum Social Choice Function onto more than two alternatives that is not a quantum dictatorship.
        \end{proof}

\section{Conclusions}
    In the field of game theory, few results have been more seminal as Arrow's Impossibility Theorem \cite{ArrowsTheorem}, a significance underscored by his Nobel Memorial Prize in Economic Sciences. Similarly, the Gibbard-Satterthwaite impossibility theorem \cite{GibbardSatterthwaite1, GibbardSatterthwaite2}, while perhaps not as universally recognized, has had profound implications at the intersection of several disciplines. Our research contributes to this rich tapestry by demonstrating that, when immersed in a quantum setting, it's possible to bypass the constraints imposed by the Gibbard-Satterthwaite impossibility theorem.  This echoes the groundbreaking work of Bao and Halpern \cite{bao2017quantumarrows}, who illustrated that Arrow's Impossibility Theorem could similarly be circumvented in a quantum setting.
    
    To show this obviation, we first formulate a definition of incentive compatibility tailored to the quantum realm, accompanied by discussion on alternative definitions. We then establish that Quantum Condorcet Voting, first introduced by Sun \textit{et al.} \cite{sun2021schrodinger}, satisfies the criteria of being incentive compatible and onto (which we also adapt for the quantum context) in addition to not being a dictatorship (as demonstrated in the original work). Finally, we prove that quantum social welfare functions, when composed with what we term the Natural Extension, yield quantum social choice functions that retain their properties of IIA and non-dictatorship. Throughout our exposition, we've endeavored to craft definitions and theorems with extendibility and wider applicability in mind and as such, provide some discussion of alternate perspectives and formulations throughout.
    
    \subsection{Future Works}
        The horizon of potential research in this domain is vast, with several avenues stemming from the discussions in the main body of this work. Here, we outline a few promising directions:

        \begin{enumerate}
            \item \textbf{Strategic Manipulation:} A deeper exploration into alternative definitions of Strategic Manipulation is warranted. Specifically, the dynamics involving strong preferences, $p$-preferences (preferences surpassing an arbitrary threshold $p \in (0, 1)$), and the extent to which a society's ballot in QCV can be influenced to align more closely with an individual voter's ballot hold promise. Some preliminary investigations in this domain have already been initiated by our team.

            \item \textbf{Quantum Voting Functions:} It would also be intriguing to study whether other quantum voting functions, such as the Quantum Majority Rule \cite{bao2017quantumarrows}, also sidestep the Quantum Gibbard-Satterthwaite conjecture. Given that the original violation of the quantum analogue of Arrow's Theorem was achieved using Quantum Majority Rule, this function might exhibit similar properties. Furthermore, there exists a significant opportunity to define additional quantum voting schemes, drawing inspiration from their classical counterparts, which may lead to intriguing results and insights.

            \item \textbf{Quantum Analogues:} Crafting quantum counterparts for other classical voting system properties and associated possibility/impossibility results, like Sen's Impossibility Theorem \cite{Sen1970impossibility} and the Muller-Satterthwaite Theorem \cite{muller1977equivalence}, could yield greater insights in this wide open field of study analogous to our findings and those by Bao and Halpern \cite{bao2017quantumarrows} and Sun \textit{et al.} \cite{sun2021schrodinger}.

            \item \textbf{Beyond Ranked Voting:} Investigating Quantum Social Choice and Welfare functions that are not merely ranked voting is another fertile ground. Systems such as Combined approval voting \cite{Felsenthal1989approvalvoting}, which allow voters to categorize each alternative as 'for', 'against', or 'neutral' might exhibit very distinct dynamics in a quantum context.
        \end{enumerate}

        In conclusion, our findings underscore the transformative potential of quantum mechanics in reshaping foundational concepts in game theory. By introducing quantum analogues and redefining classical axioms, we not only challenge established impossibility theorems but also pave the way for a richer understanding of strategic interactions in complex systems. As the field of quantum game theory continues to evolve, we anticipate that it will offer novel solutions to long-standing challenges and further bridge the gap between classical and quantum decision-making.

   %
   %
   %

\section*{Acknowledgements}
    We would like to thank Dr. Alex Psomas for his constructive feedback and guidance.

\newpage
\begin{appendices}
\section{Additional Readings}\label{app_lit_review}

%

    
    \subsection{Brief Introduction to the History of Quantum Computing}
        Quantum computing has been the subject of increased excitement and study since the discovery of Shor's fast factoring and discrete logarithm algorithms \cite{Shor_1997}. This has lead to what appears to be a significant speedup over classical systems in certain tasks. One such task is solving the \textit{hidden subgroup problem}, which is closely connected with the discovery of the Quantum Fourier Transform and its fundamental impact in quantum computing theory. Another is \textit{quantum search algorithms}, which have increasing interest due not to their incredible speedup (merely quadratic), but rather to their very wide set of known and potential applications. A third such task is \textit{simulation of quantum systems}, a highly promising field that is already showing benefits in the Noisy Intermediate Stage Quantum computing era (NISQ era) \cite{MikeAndIke2011, olson2017quantum, Preskill2018quantumcomputingNISQera}.
        
        Despite the incredible amount of research published in the field in the last four decades or so, there are still very large unanswered questions. Some examples of these questions are exactly how powerful quantum computers are, whether they are more powerful than classical computers, and where their power comes from. The literature has yet to determine concrete answers to these questions, though fast factoring and other algorithms give strong evidence to their power and advantage. Regarding the power of quantum computation over classical computation, it has been shown that if we prove quantum computation is more powerful than classical computation then \textsc{P} is not equal to \textsc{PSPACE}. There have been many attempts to prove this relationship classically, so it seems reasonable to assume that proving quantum computation is more powerful than classical computation will be difficult.
    
        One important application that has only recently begun to be explored in the past couple decades is quantum game theory. This has many branches, from quantum cryptography and encryption to quantum games to pure game theory results involving Nash equilibria. We give a brief survey of each of these fields following a different perspective on quantum computing.
    
        \subsubsection{Statistical Perspective} One perspective we can take on quantum computing is from statistics. The \textit{output} of a quantum circuit (a quantum state) can be thought of as a joint probability distribution over the qubits. The primary advantage that quantum computers give is the ability to efficiently perform certain manipulations (some of which are thought to be classically intractable \cite{gottesman1998heisenberg, aaronsonGottesman2004ImprovedSimulation}) on the distribution in order to achieve desired end states. This provides the core of the \textsc{SampBQP} - the class of sampling problems efficiently solvable on a quantum computer \cite{Lund2017}. It is thought that sampling from certain quantum probability distributions classically (such as the problems in \textsc{SampBQP}) requires exponentially scaling resources \cite{Boixo2018, Aaronson_2005, Bremner_Jozsa_Shepherd_2010, Bouland2019, aaronson2011ComplexityLinearOptics, Fujii_Morimae_2017, Bremner2016AvgCase}. This line of thinking has also been the basis of at least one quantum supremacy demonstration on the basis of random quantum circuit sampling \cite{Arute2019}.
        
        For the sake of accuracy, it is important to note that while the final measurement of a quantum circuit acts similarly to sampling from a joint probability distribution, in reality it is sampling from a circuit that, at every point, is some form of an \textit{imaginary} joint probability distribution \cite{WOOLNOUGH2023}. The imaginary part allows for new ways of \textit{interfering} (strategically manipulating your distribution) to expose certain problem properties. A quantum circuit utilizes quantum mechanics to build such distributions, which appears to allow for more efficient manipulations than classical computers.
    
    \subsection{Quantum Cryptography and Encryption}
        The introduction of quantum computing and Shor's algorithm allow for fundamental attacks on classical encryption methods derived from computational complexity, particularly the integer factorization problem (or the discrete logarithm problem, etc.), and schemes based on these problems. However, quantum mechanics also allows for new methods of cryptography, introduced as early as the 1970's and 80's \cite{wiesner1983conjugate, Bennett1985Crypto} and continue to be be utilized through present-day (a simple arXiv or Google Scholar search for quantum cryptography in the last 2 years yields thousands of results developing and implementing previous protocols and introducing new ones). These protocols usually involve using a combination of quantum and classical information sharing and have many potential advantages involving the ability to consistently detect eavesdroppers, produce messages that are resistant to single-channel attacks, and perform \textit{provably} secure key distribution over public channels \cite{Gisin2002Crypto, MikeAndIke2011}. Such schemes can then be used in many places where classical cryptography is used, such as blockchain and the internet.  
    
    \subsubsection{Quantum Blockchain} 

        An additional impact of the rising prominence of quantum computing lies in its profound implications for the field of blockchain. Blockchain is a form of distributed ledger with high tolerance for unreliable actors. This technology is largely based on digital signatures and cryptographic hash functions. Classical digital signatures can be attacked via quantum computing \cite{Schneier_1994}, which has wide reaching implications within blockchain. The security of blockchain schemes can be increased by using Post-Quantum Cryptography \cite{Bernstein2009}, a very actively researched field which deals with ``quantum-safe" classical cryptographic algorithms. Some, however, view the assumptions of post-quantum cryptography as unproven \cite{Kiktenko_2018}, and propose specifications for a quantum blockchain that utilizes quantum key distribution, guaranteeing information-theoretic security through use of well-proven and well-studied quantum cryptography results \cite{Kiktenko_2018, sun2018quantum}. Additionally, quantum blockchain is intrinsically linked to the idea of quantum voting, due to the distributed decision making processes within it, and there have been works creating voting schemes utilizing quantum blockchain as well \cite{Sun2019QBlockchain}.

    \subsection{Quantum Game Theory}

        As mentioned in the previous section, the introduction of quantum ideas has had profound effects on the domain of Game Theory, spawning an entirely new field. This new field adds ideas such as superposition, entanglement, and interference to traditional game theory. Quantum Game Theory, though still in its formative years, promises to reshape our understanding of strategic interactions by leveraging the counterintuitive principles of quantum mechanics \cite{flitney2002introduction}.
        
        Classical game theory, with its roots in economics, mathematics, and social science, provides a structured framework to analyze strategic interactions among rational decision-makers. It revolves around concepts like Nash equilibria, where players' strategies are mutually best responses to each other, and theorems like Gibbard-Satterthwaite, which delve into the impossibilities of certain fair decision-making processes. However, the classical framework, with its deterministic and static nature, has its limitations \cite{Ichikawa_2008}.
        
        Enter Quantum Game Theory. Beyond just introducing quantum versions of classical games, it redefines the very theoretical underpinnings of strategic decision-making. The foundational shift from classical to quantum in game theory is not merely a change in mathematical formalism. It introduces a richer strategic space where players can harness quantum phenomena to their advantage \cite{HANAUSKE2007650}. Superposition allows for strategies to exist in a combination of states, offering a more fluid and dynamic strategic landscape. Entanglement introduces correlated strategies between players that transcend classical bounds, leading to deeper cooperative and competitive dynamics\cite{flitney2002introduction}.

        In the subsequent subsections, we offer a concise overview of select works in Quantum Game Theory. We will touch upon the core principles that define quantum games, providing readers with an understanding of their unique characteristics. Following this, we will briefly discuss the interplay between Nash equilibrium and notable quantum game theory results that have emerged in recent literature. Lastly, we will highlight some insights from quantum voting systems and their intriguing connection to Arrow's theorem. While our review is by no means exhaustive, it aims to provide a snapshot of the current landscape and set the stage for our primary focus: bypassing the Gibbard-Satterthwaite theorem in a quantum context.
    
        \subsubsection{Quantum Games}
        
        A closely related field (or even subfield) to quantum game theory is quantum games. Quantum games take normal games and allow for a much richer strategic landscape by incorporating superposition, entanglement, and randomness of quantum measurements into any and all parts of new and old games. A simple example motivated by classical game theory is the quantum Prisoners' Dilemma, which has been shown to allow the prisoners to ``escape" the dilemma if quantum strategies are used \cite{eisert1999quantum}. From this seminal work, a field of quantum game theory quickly blossomed, and has been applied to other classical game theory problems including Battle of the Sexes, the Monty-Hall problem, and many others \cite{flitney2002introduction}. Quantum game theory has also opened the door to many new areas such as quantum combinatorial games, which are versions of classical combinatorial games enriched by allowing quantum moves \cite{burke2020quantum}. While there is a plethora of research already in this domain, many areas of game theory are just starting to be explored through the quantum lens.
        
        \subsubsection{Nash Equilibrium and Quantum Game Theory Results}

        Getting into the more theoretic aspects of game theory, there are several results indicating that quantum computing allows better or more efficient calculations of solutions in topics such as Nash equilibrium, online learning algorithms, game strategies and representations, and more.  In Nash equilibrium, there have been many recent results in using online quantum algorithms to get $\varepsilon$-approximate Nash equilibrium of zero-sum games, up to and including a zero-regret result with cost $\Tilde{O}(\sqrt{m+n}/\varepsilon^{2.5})$ \cite{gao2023logarithmic, vanapeldoorn_et_al2019, bouland2023quantum}. Gao \textit{et al.} provide a good summary of the work in this area \cite{gao2023logarithmic}. Good reviews of the field from a more game-theoretical perspective can be found from Piotrowski and Sladkowski \cite{Piotrowski2003} or Flitney and Abbott \cite{flitney2002introduction}, as well as a number of other further readings \cite{landsburg2004quantum, piotrowski2003stage}.

        \subsubsection{Quantum Voting Systems and Arrow's Theorem} 
        A particularly interesting result in quantum game theory is the bypassing of a formulation of Arrow's Theorem in a quantum setting. Arrow's Theorem is a seminal impossibility result in game theory that states that no social welfare function with more than three alternatives can hold the properties of unanimity and independence of irrelevant alternatives (IIA) without being a dictatorship \cite{ArrowsTheorem}. Intuitively, the United States voting system is not a social welfare function and thus this result does not apply to it; while there is another impossibility result, the Gibbard-Satterthwaite (GS) Theorem \cite{GibbardSatterthwaite1, GibbardSatterthwaite2}, that gives a similar impossibility result for (ordinal) social choice functions (while the US voting system \textit{is} a social choice function, it is not a ordinal social choice function), the US voting system gets around being a dictatorship (according to the theorem) by violating another one of its properties, namely incentive compatibility (the idea that if you know your choice alternate $A$ will never win, it is not beneficial for you to switch to your second-favorite alternate $B$ so that $B$ will win, and thus misrepresent your true choice).
        
        Originally shown by Bao and Halpern \cite{bao2017quantumarrows}, reformulation of the voting system and properties such as unanimity, IIA, and dictatorship into a quantum setting allows for certain quantum voting systems to bypass a quantum Arrow's Theorem. Bao and Halpern \cite{bao2017quantumarrows} showed that a Quantum Majority Rule bypasses the quantum version and Sun \textit{et al.} \cite{sun2021schrodinger} showed that a Quantum Condorcet Voting system bypasses the quantum version similarly. Sun \textit{et al.} additionally gives alternative definitions of certain quantum properties, with more nuanced proofs. Voting systems are used in a large variety of useful settings outside political voting, such as auctions (for humans or for ads), blockchain, IoT, etc.
        Further review of quantum voting systems is provided after a high-level discussion of quantum vs probabilistic voting systems.
    
    \subsection{Quantum vs Probabilistic Voting Systems}


        \subsubsection{Probabilistic Voting Systems}
        A natural connection to quantum voting systems is probabilistic voting systems, due to the appearance of similar outputs. Probabilistic voting systems in relation to the GS impossibility theorem were first studied by Gibbard himself in 1977, where he showed that strategy proof (SP) randomized voting rules that are decision schemes are probability mixtures over deterministic rules, each of which are either a dictatorship or can only produce 2 alternates \cite{Gibbard1977Random, Procaccia_2010}. Working within the class of SP randomized voting rules, Procaccia \cite{Procaccia_2010} showed that using certain SP randomized rules, one can achieve a $\gamma$-approximation of the deterministic version of those rules. Though most of the $\gamma$s Procaccia proved were not very optimal, they provided a solid foundation to begin studying other rules in the same context. In relation to SP-ness, BBF \cite{brandl2015incentives}, and later ALR \cite{Aziz_Luo_Rizkallah_2018}, showed that there are reasonable randomized voting rules that satisfy \textit{strong participation}, where participation in the vote strictly increases utility for an agent. This notion of participation is based on stochastic dominance (SD), which is likely the most popular approach for defining how resulting probability mixtures are compared \cite{lederer2021non} (for example, voter utility, which is particularly useful when discussing SP-ness). Regardless, such notions of SP-ness all require ordinal preferences as input, which does not make as much sense in a quantum setting. An additional notion of truthfulness (which thus implies a notion of SP-ness) is \textit{lex-truthfulness}, introduced by Chakrabarty and Swamy \cite{Chakrabarty2014welfare}, a stronger notion of truthfulness than SD-SP-ness which allows for a class of voting mechanisms that bypass GS. While there are other notions that do the same, lex-truthfulness more directly deals with probability mixtures over rankings. Finally, a more useful notion discussed by Brandt \cite{brandt2017rolling} is PC-SP-ness, where a voter prefers probability mixture $a$ over $b$ if $a$ is more likely to return a better outcome. This notion of SP-ness is more interesting from a quantum perspective, as it more intuitively deals with voter utility in probabilistic contexts and can better handle a probabilistic mixture over rankings instead of over candidates.
        
        Another important set of rules that bypass the GS results are \textit{approximate} SP randomized rules with approximation ratio $\varepsilon$. Birrell and Pass \cite{birrell2011approximately} showed that the existence of non-trivial approximate SP randomized voting rules depended very heavily on $\varepsilon$. They showed that $\varepsilon = \omega(1/n)$ allowed for natural sets of these rules, but that $\varepsilon$ as strong as $o(1/n^2)$ prevented existence of such rules. Their measure of ``closeness" is additionally more nuanced than standard definitions, but is nevertheless natural.

        \subsubsection{Quantum Voting Systems are Not Strictly Probabilistic} \label{sec_intro_probvsquant}
        All of the aforementioned research focuses on randomized voting rules that take classical ordinal preferences as input. While it is important to discuss such results, they do not apply in a quantum setting for two primary reasons. First, quantum voting systems provide a rich setting for voter preferences, allowing for both preferences and ballots which are entangled or are in superposition. In this context, a ballot only in superposition can be understood as a probability distribution over all possible total ordinal preferences. Alternatively, it can be understood as a joint probability distribution over alternatives. For example, if the selection of one alternate for the top position prevents selection of another alternate for the second position. Two (or more) ballots which are entangled allows each voter involved to coordinate their votes. For example, if agent $a$ and $b$ decided to submit identical probabilistic ballots, if $a$'s ballot is measured (sampled) to be alternate 1, agent $b$'s ballot will also produce alternate 1 with 100\% probability despite the original distribution. In short, entanglement allows for voting along a party line (or otherwise correlated with one) and superposition allows for split-decision preferences.
        
        The second reason these SP randomized voting results do not apply in a quantum system is that they do not consider joint probability distributions in the output. Typically, studied voting systems deal directly with probability mixtures over individual alternates, where joint probability mixtures allow for probability mixtures over \textit{total ordinal rankings}. A trivial example of this would be president-vice president pairings. Given presidential pair $a$ and $b$, and pair $c$ and $d$, if all are considered in a single ordinal preference, a joint probability mixture could allow for \textit{either} $a$ in position 1 and $b$ in position 2 \textit{or} $c$ in position 1 and $d$ in position 2, but not a situation where $a$ is in position 1 and $d$ is in position 2 (or vice versa). 
        An additional consideration is the lack of an inherently quantum processes in SP randomized voting rules. This, however, does not end up being relevant from a purely theoretic standpoint as long as the input and output to the voting system remain the same (i.e. simple ordinal preferences as input, probabilistic mixture over alternates as output).

        \subsubsection{Justification for \textit{Quantum} Voting Systems}
        One might object that quantum settings are not intuitive for a voter to understand, and thus we should restrict notions such as truthfulness and utility to how they are represented in either classical or traditional probabilistic settings, as a practical matter. While the intrinsic intellectual merits of quantum voting are particularly intriguing, it is also not natural, practically, to restrict such a powerful domain to a classical one, even if it is probabilistic. Quantum computers can simulate all classical operations efficiently \cite{MikeAndIke2011}, but even traditional probabilistic settings cannot mimic and do not consider foundational quantum effects such as entanglement and superposition, respectively. Thus, Bao and Halpern \cite{bao2017quantumarrows} reformulated Arrow's Theorem, voting systems, and desirable voting properties to allow for quantum effects. In this new setting, it is an open problem whether all probabilistic voting systems can bypass the reformulated Arrow's Theorem, but it would be trivial to show that two can. In particular, the voting systems proposed by Bao and Halpern, and Sun \textit{et al.}, can be reformulated as probabilistic voting systems that produce joint probability distributions over the alternate rankings, while keeping the results of \cite{bao2017quantumarrows, sun2021schrodinger} and bypassing the quantum formulation of Arrow's Theorem.
    
        An interesting avenue of research to pursue would be studying probabilistic game theory results in a quantum context. To the authors' knowledge, there are no results yet published in this area. Quantum computing could contribute to the field of probabilistic game theory by indicating which probability distributions (or operations on them) are more interesting to study, due to their resource efficiency using a quantum circuit (both space and time), or by providing a more restricted/structured formalism to study probabilistic results with. In particular, steps performed in a voting system on a quantum computer deal each step with \textit{imaginary} joint probability distributions and get to use the full power of quantum mechanics when manipulating them. This allows for new ways of manipulation including \textit{interference}, a powerful quantum-mechanical tool for amplifying groups of or individual probability amplitudes.

    \subsection{Voting Protocols in a Quantum Setting}

        Returning to voting protocols, there are not many that have been created in a quantum setting. Sun \textit{et al.} \cite{sun2023LogicOperators} gave appropriate distinctions between \textit{quantum secured voting} and \textit{quantum computed voting}, with the former representing classical voting secured by quantum protocols and the latter representing fundamentally quantum voting. We will focus on quantum computed voting protocols, which take a more social choice theory perspective on quantum voting.
        
        Most of these quantum computed voting protocols are focused on binary decisions and ensuring security or unanimity for their voters. For example, Vaccaro \textit{et al.} \cite{vaccaro2007AnonymousVote} have formulated protocols for anonymous voting and surveying of binary choices such as comparative ballot, which is a protocol where neither party can know how the other voted beforehand, and the tally-man can only determine if they voted identically. Similarly, they show a protocol for anonymous surveying such that the tallyman cannot determine the tally before the end, and the voters acting on the quantum state similarly cannot determine the current count. In a similar vein, Sun \textit{et al.} devised a voting method based on quantum blockchain utilizing committed ballots for making binary decisions anonymously \cite{Sun2019QBlockchain}. In terms of ranked voting over alternatives as utilized later in this work, two major works are those of Bao and Halpern with Quantum Majority Rule \cite{bao2017quantumarrows} and Sun \textit{et al.} with Quantum Condorcet Voting \cite{sun2021schrodinger}, both of which were shown to bypass their crafted quantum analogues of Arrow's Theorem, as previously discussed.
    
        Finally, in an attempt to make quantum voting protocols more practical, Sun \textit{et al.} \cite{sun2023LogicOperators} introduced quantum building blocks for quantum voting protocols. Using quantum voting rules termed quantum logical veto and quantum logical nomination, they formulate a quantum logical majority vote, quantum logical Condorcet vote, and finally a quantum average rule. These building blocks present an interesting opportunity to study more practical implementations of quantum voting protocols in a formulaic and structured way, and allow quantum voting rules and protocols to be compared and contrasted directly with classical voting rules and protocols.


    \section{More Definitions and Theorems}\label{app_defns}
        \subsection{Voting Systems}
            \begin{definition}[Rankings \cite{alggametheorybook}]\label{def_ranking}
                Let $A$ be the set of alternatives and $\mathcal{V}$ be the set of voters. We define the set of rankings $\mathcal{L}(A)$ as the set of linear (total) orderings of $A$, and $R_i \in \mathcal{L}(A)$ as the ranking by $v_i$. This ranking is sometimes also referred to as $\prec_{i}$, used in the sense of $a_k \prec_{i} a_j$, saying that alternative $j$ is preferred to alternative $k$ in ranking $i$.
            \end{definition}

        \subsubsection{Classical Social Choice Functions (SCFs) and Social Welfare Functions (SWFs)}
            The results of a classical voting system determine two classes of functions, one for the case of a single winner selected from the alternatives, and one for a total ranking determined by each voters ranking.
                \begin{definition}[Classical Social Choice Function (SCF) \cite{alggametheorybook}]\label{def_SCF}
                    A \textit{Classical Social Choice Function} is a function $f: \mathcal{L}(A)^{n} \to A$. This takes the ranked preferences of each voter, and determines a single victor from among the set of alternatives.
                \end{definition}
    
                \begin{definition}[Classical Social Welfare Function (SWF) \cite{alggametheorybook}]\label{def_SWF}
                    A \textit{Classical Social Welfare Function} is a function $F: \mathcal{L}(A)^{n} \to \mathcal{L}(A)$. This takes the ranked preferences of each voter, and determines a final ranking from the set of possible rankings.
                \end{definition}

            Each of these has different associated results, including Arrow's Theorem and Gibbard-Shatterwaite, though there is a connection between them classically \cite{alggametheorybook}.
        \subsection{Arrow's and Properties}
            \begin{theorem}[Arrow's Impossibility Theorem \cite{alggametheorybook}]\label{thm_arrows}
                Every SWF over a set of more than 2 alternatives $(|A| > 2)$ that satisfies Unanimity and Independence of Irrelevant Alternatives is a Dictatorship.
            \end{theorem}

            \subsubsection{Unanimity} 
            First, we will discuss the axiom of Unanimity. Classically, Unanimity is defined as follows.
                \begin{definition}[Unanimity \cite{alggametheorybook}]\label{def_unanimity}
                    An SWF $F$ is \textit{Unanimous} (or has the property of Unanimity) if for two arbitrary alternatives $a_j, a_k \in A$ and for $\prec ~ =  F(R_1, \dots, R_n)$,
                    $$(\forall v_i \in \mathcal(V), a_j \prec_i a_k) \Rightarrow a_j \prec a_k.$$
                \end{definition}

            In more colloquial terms, unanimity is the idea that if all voters agree on the relative positioning of two alternatives, the final ranking should conform to that relative positioning. The reasons this property is desirable are intuitive.

            This naturally leads to a formulation of a quantum analogue. In Sun \textit{et al.} \cite{sun2021schrodinger}, it is defined as follows.
            
                \begin{definition}[Quantum Unanimity \cite{sun2021schrodinger}]\label{def_q_unanimity}
                    A QSWF $\mathcal{E}$ is \textit{Sharply Unanimous} (or has the property of Sharp Unanimity) if it satisfies the following:
                    \begin{itemize}
                        \item For all quantum societal ballot profiles $\rho$ and all pairs of alternatives ($a$, $a'$), if Tr$(\Pi^{a \succ a'}(\rho_i)) = 1$ for each voter $v_i$, then Tr$(\Pi^{a \succ a'}(\mathcal{E}(\rho))) = 1$
                    \end{itemize}
                    
                    A QSWF $\mathcal{E}$ is \textit{Unsharply Unanimous} (or has the property of Unsharp Unanimity) if it satisfies the following:
                    \begin{itemize}
                        \item For all quantum societal ballots $\rho = \mathcal{E}(\rho_1 \otimes \dots \otimes \rho_n)$ and all pairs of alternatives ($a$, $a'$), if Tr$(\Pi^{x \succ y}(\rho_i)) > 0$ for each voter $v_i$, then Tr$(\Pi^{x \succ y}(\mathcal{E}(\rho))) > 0$
                    \end{itemize}
    
                    A QSWF $\mathcal{E}$ satisfies the quantum unanimity condition if it satisfies both sharp and unsharp unanimity conditions.
                \end{definition}
            
            Otherwise stated, $\mathcal{E}$ is \textit{Sharply Unanimous} if it follows that if each voter individually ranked x above y with probability 1, then the social welfare function will also rank x above y with probability 1.  The same interpretation holds for \textit{Unsharp Unanimity}, with nonzero probabilities in exchange for probabilities of 1.
        
        \subsubsection{Independence of Irrelevant Alternatives (IIA)} 
            Second, we will discuss the axiom of the Independence of Irrelevant Alternatives. Classically, the axiom of the Independence of Irrelevant Alternatives is defined as follows.
             
                \begin{definition}[Independence of Irrelevant Alternatives (IIA) \cite{alggametheorybook}]\label{def_iia}
                    An SWF $F$ is \textit{Independent of Irrelevant Alternatives} (or has the property of Independence of Irrelevant Alternatives) if for two arbitrary alternatives $a_j, a_k \in A$ and where $F(R_1, \dots, R_i, \dots, R_n) = \prec$ and $F(R_1', \dots, R_i', \dots, R_n') = \prec'$
                    $$\forall i(a_j \prec_i a_k \Leftrightarrow a_j \prec_i' a_k) \Rightarrow (a_j \prec a_k \Leftrightarrow a_j \prec' a_k).$$
                \end{definition}

            In other words, if the relative rankings of two alternatives don't change for each voter, the relative rankings of those two alternatives should stay the same in the results from those regardless of the existence or movement of other alternatives between the two sets of rankings.

            Sun \textit{et al.} again provide a definition of this as follows.
            
                \begin{definition}[Quantum Independence of Irrelevant Alternatives \cite{sun2021schrodinger}]\label{def_qiia}
                    A QSWF $\mathcal{E}$ satisfies the \textit{Sharp IIA} condition if it satisfies the following:
                    \begin{itemize}
                        \item For all quantum ballot profiles $\rho$ and $\rho'$ and all pairs of alternatives ($a$, $a'$), if Tr$(\Pi^{a \succ a'}(\rho_i))$ = Tr$(\Pi^{a \succ a'}(\rho_i' ))$ for each voter $v_i$, then Tr$(\Pi^{a \succ a'}(\mathcal{E}(\rho))) = 1$ implies that Tr$(\Pi^{a \succ a'}\mathcal{E}(\rho'))) = 1$
                    \end{itemize}
    
                    A QSWF $\mathcal{E}$ satisfies the \textit{Unsharp IIA} condition if it satisfies the following:
                    \begin{itemize}
                        \item For all quantum ballot profiles $\rho $ and $\rho'$ and all pairs of alternatives ($a$, $a'$), if Tr$(\Pi^{a \succ a'}(\rho_i))$ = Tr$(\Pi^{a \succ a'}(\rho_i' ))$ for each voter $v_i$, then Tr$(\Pi^{a \succ a'}(\mathcal{E}(\rho))) > 0$ implies that Tr$(\Pi^{a \succ a'}(\mathcal{E}(\rho'))) > 0$
                    \end{itemize}
                \end{definition}

            Otherwise stated, if two different quantum ballot profiles both have support (with probabilities 1 or $>$0 for sharp and unsharp, respectively) for a particular ranking $x \succ y$, then the two quantum ballot profiles should reflect very similar support in their transformed rankings.

            \subsubsection{Dictatorship}
            Third, we will discuss the axiom of Dictatorship. Classically, the axiom of Dictatorship is defined as follows.
                \begin{definition}[Dictatorship \cite{alggametheorybook}]\label{def_dict}
                    Voter $i$ is a \textit{dictator} in social welfare function $F$ if $\forall \prec_1 \dots \prec_n \in L(A)$, $F(\prec_1 \dots \prec_n) = \prec_i$. That is, the societal preference in a dictatorship is that of the dictator.
                \end{definition}

            In other words, the societal output is completely determined by a single voter's ballot.

            Sun \textit{et al.} again provide a definition of this as follows.
            
                \begin{definition}[Quantum Dictatorship \cite{sun2021schrodinger}]\label{def_quantum_dict}
                    A QSWF $\mathcal{E}$ satisfies sharp dictatorship if there is a voter $v_i$ such that:
                    \begin{itemize}
                        \item For all quantum ballot profiles $\rho = (\rho_1,\dots,\rho_n)$ and all pairs of candidates ($x$, $y$), \newline
                        Tr$(\Pi^{x\succ y} \rho_i) = 1 \Leftrightarrow$ Tr$(\Pi^{x \succ y}(\mathcal{E}(\rho))) = 1$.
                    \end{itemize}
                    A QSWF $\mathcal{E}$ satisfies unsharp dictatorship if there is a voter $v_i$ such that:
                    \begin{itemize}
                        \item For all quantum ballot profiles $\rho = (\rho_1,\dots,\rho_n)$ and all pairs of candidates ($x$, $y$), \newline
                        Tr$(\Pi^{x\succ y} \rho_i) > 0 \Leftrightarrow$ Tr$(\Pi^{x \succ y}(\mathcal{E}(\rho))) > 0$.
                    \end{itemize}
                    A QSWF $\mathcal{E}$ satisfies quantum dictatorship if it satisfies both sharp and unsharp dictatorship.
                \end{definition}

            Sharp dictatorship states that whenever the dictator prefers $x$ to $y$ with certainty, then so does the society. Unsharp dictatorship states that whenever the dictator prefers $x$ to $y$ with positive probability, then so does the society.            

        \subsection{Quantum analogue of Arrow's Theorem}
            As defined by Sun \textit{et al.},
            \begin{conjecture}[Quantum analogue of Arrow's Theorem \cite{sun2021schrodinger}]\label{conj_q_arrows}
                Every QSWF over a set of more than 2 alternatives $(|A| > 2)$ that satisfies quantum unanimity (both sharp and unsharp) and quantum independence of irrelevant alternatives (both sharp and unsharp) is a quantum dictatorship (either sharp, unsharp, or both).
            \end{conjecture}

    \section{Reference Table} \label{app_ref_table}

    \begin{tabular}{c|c|c}
        Term                                                            & Symbol/Acronym                                                                                        & Notes/Additional Info\\ \hline
        Alternative                                                     & $a \in A$                                                                                             & $|A| = m$, what is being voted on\\ \hline
        Voter                                                           & $v_i$                                                                                                 & $n$ voters, distinguished by index $i$. \\ \hline
        Ranking                                                         & $R \in \mathcal{L}(A), \mathcal{L}$                                                                      & \shortstack{When used with a subscript, it refers \\ to the ranking submitted by voter $i$.} \\ \hline
        Preference                                                      & $a \succ a'$                                                                                          & \shortstack{This is interpreted as preferring \\ alternative $a$ to alternative $a'$.}\\ \hline
        Ballot                                                          & \shortstack{$\prec_i$\\$\rho_i$}                                                                      & \shortstack{The ballots submitted by voters.\\ Classical and  Quantum, respectively.} \\ \hline
        Ballot Profile                                                  & \shortstack{$R = (R_1, \dots, R_n)$, \\ $\rho = \rho_1 \otimes \dots \otimes \rho_n$}       & \shortstack{The ballot profiles submitted by voters.\\ Classical and Quantum, respectively.} \\ \hline
        \shortstack{Classical Social \\ Choice Function(s)}             & SCF(s)                                                                                                & $f: \mathcal{L}(A)^{n} \to A$ \\ \hline 
        \shortstack{Classical Welfare \\ Choice Function(s)}            & SWF(s)                                                                                                & $F: \mathcal{L}(A)^{n} \to \mathcal{L}(A)$ \\ \hline
        \shortstack{Hilbert Space of\\ Rankings}                        & $\mathfrak{R}, \mathfrak{R}_i$                                                                        & $\mathbb{C}^{|\mathcal{L}(A)|}$ \\ \hline
        \shortstack{Hilbert Space of\\ Alternatives}                    & $\mathcal{A}$                                                                                         & $\mathbb{C}^{|A|}$ \\ \hline
        Preference Subspace                                             & $\mathcal{S}^{a \succ a'}$                                                                            & \shortstack{Subspace of the Hilbert Space of Rankings that has \\ the preference $a \succ a'$ as part of its basis vectors.} \\ \hline
        \shortstack{Winning Alternative \\ Subspace}                    & $\mathcal{S}^{a}$                                                                                     & \shortstack{Subspace of the Hilbert Space of Rankings that has \\ $a$ as highest ranked alternative for its basis vectors.} \\ \hline
        \shortstack{Projection onto \\ Preference Subspace}             & $\Pi^{a \succ a'}(\rho)$                                                                              & \shortstack{Projection Operator of ballot $\rho$ \\ onto the subspace $\mathcal{S}^{a \succ a'}$.} \\ \hline
        \shortstack{Projection onto \\ Alternative Subspace}            & $\Pi^{a}(\rho)$                                                                                       & \shortstack{Projection Operator of ballot $\rho$ \\ onto the subspace $\mathcal{S}^{a}$.} \\ \hline
        \shortstack{Probability of a\\preference (ranking)}             & Tr$(\Pi^{a \succ a'}(\rho))$                                                                          & Probability assigned to $a \succ a'$ in the ballot $\rho$. \\ \hline
        \shortstack{Probability of a\\preference (choice)}              & Tr$(\Pi^{a}(\rho))$                                                                                   & \shortstack{Probability assigned to $a$ being the highest \\ ranked alternative in the ballot $\rho$.} \\ \hline
        \shortstack{Quantum Welfare \\ Choice Function(s)}              & QSWF(s),  $\mathcal{E}$                                                                                & $\mathcal{E}: D(\mathfrak{R}_1 \otimes \dots \otimes \mathfrak{R}_n) \to D(\mathfrak{R})$ \\ \hline      
        \shortstack{Societal Ranking \\ Density Operator}               & $\rho_{soc} \in D(\mathfrak{R})$                                                                      & \shortstack{A Density Operator on Hilbert Space of $\mathfrak{R}$ \\ usually as the output of a QSWF.} \\ \hline
        \shortstack{Quantum Social \\ Choice Function(s)}               & QSCF(s), $\xi$                                                                                        & $\xi: D(\mathfrak{R}_1 \otimes \dots \otimes \mathfrak{R}_n) \to D(\mathcal{A})$ \\ \hline 
        \shortstack{Quantum Social \\ Choice Extension(s)}              & QSCE(s)                                                                                               & \shortstack{$h: D(\mathfrak{R}) \to D(\mathcal{A})$.\\ Composition with QSWF gives QSCF.} \\ \hline 
        \shortstack{Societal Alternative \\ Density Operator}           & $\alpha \equiv \alpha_{soc} \in D(\mathcal{A})$                                                       & \shortstack{A Density Operator on $\mathcal{A}$ \\ usually as the output of a QSCF.} \\ \hline
        \shortstack{Quantum Condorcet \\ Voting}                        & \shortstack{QCV \\ QCVNE}                                                                                                   & \shortstack{Specific QSWF originally defined in \\ Sun \textit{et al.} \cite{sun2021schrodinger}. QCV with the \\ Natural Extension} \\ \hline
        \shortstack{Natural Quantum Social \\ Choice Extension(s)}      & NQSCE, $\Lambda$                                                                                             & \shortstack{Specific QSCE.  Projects density of ranking \\ to their highest ranked alternative. \\ Sometimes called the Natural Extension.} \\ \hline
        \shortstack{Quantum Incentive \\ Compatible}                    & QIC                                                                                                   & Essentially, not Quantum Strategically Manipulatable \\ \hline

    \end{tabular}
\end{appendices}

\newpage
\bibliographystyle{acm}
\bibliography{citations}

\end{document}